\documentclass[a4paper,10pt]{article}
\usepackage[english]{groreport}
\usepackage{graphicx}
\usepackage{amssymb}
\usepackage{amsmath}
\usepackage{mathrsfs}
\usepackage{mathabx} % a voir
\usepackage{shuffle}
\usepackage{algorithm}
\usepackage{algorithmic}
\usepackage{booktabs} % for tables
\usepackage{wrapfig}
\usepackage{listings} % to list code
\usepackage{xcolor}
\usepackage{color}
\usepackage{comment}

\usepackage{fancyhdr}
\usepackage[latin1]{inputenc} %pour avoir des beaux pdf
\usepackage{placeins} %utiliser avec \FloatBarrier
\usepackage{float}
\usepackage{floatflt}
\usepackage{subfloat}
\usepackage{theorem}
\usepackage{relsize}

\usepackage{caption}
\usepackage{subcaption}
\usepackage{epstopdf}
\usepackage{QFHeader}

\usepackage{bbm}

\usepackage{pgfplots}
\usepackage{tikz}
\usepackage{tikzscale}
\usetikzlibrary{shapes.multipart}
\usetikzlibrary{shapes.geometric, arrows, calc, automata, petri, shapes, decorations.pathmorphing, decorations.pathreplacing, positioning, fit, backgrounds}

\newtheorem{definition}{\textbf{Definition}}[section]
 
\newtheorem{theorem}{\textbf{Theorem}}[section]

\newtheorem{proposition}{\textsl{\textbf{Proposition}}}%[part]
\newtheorem{lemma}{\textbf{Lemma}}[section]
 
\newtheorem{remark}{\textsl{\textbf{Remark}}}[section]

\newtheorem{proof}{\textbf{proof}}[section]
\newcommand{\be}{\begin{equation}}
\newcommand{\en}{\end{equation}}
\newcommand{\beq}{\begin{equation}}
\newcommand{\eeq}{\end{equation}}
\newcommand{\bes}{\begin{eqnarray*}}
\newcommand{\ens}{\end{eqnarray*}}
\newcommand{\beqa}{\begin{eqnarray*}}
\newcommand{\eeqa}{\end{eqnarray*}}
\newcommand{\bea}{\begin{eqnarray}}
\newcommand{\ena}{\end{eqnarray}}

\newcommand{\xuparrow}[1]{%
  {\left\uparrow\vbox to #1{}\right.\kern-\nulldelimiterspace}
}

\newcommand{\xdownarrow}[1]{%
  {\left\downarrow\vbox to #1{}\right.\kern-\nulldelimiterspace}
}

\DeclareMathAlphabet{\mathsfit}{\encodingdefault}{\sfdefault}{m}{sl}
\SetMathAlphabet{\mathsfit}{bold}{\encodingdefault}{\sfdefault}{bx}{n}

\title{\textbf{The Geometry of Risk: Path-Dependent Regulation and Anticipatory Hedging via the SigSwap} 
} 

\author{Daniel \textsc{Bloch}  
\footnote{Visting Professor at the College of Engineering and Computer Science, VinUniversity, Hanoi.} \\
\textsc{University of Paris 6 \& VinUniversity} \\
Working Paper \\
\footnote{All mistakes are ours.} 
}
\date{19th of March 2026 \\
Version : 1.0.0}

\begin{document}
\QFHeader{\textbf{The Geometry of Risk: Path-Dependent Regulation and Anticipatory Hedging via the SigSwap}}{Daniel Bloch}{19th of March 2026}
\maketitle

\begin{abstract}
This paper introduces a transformative framework for managing path-dependent financial risk by shifting from traditional distribution-centric models to a geometry-based approach. We propose the SigSwap as a new regulatory instrument that allows market participants to decompose complex risk into terminal price law and the underlying texture of the price path. By utilising the mathematical properties of the path-signature, we demonstrate how previously unmodellable risks, such as lead-lag dynamics and flash-crash spiralling, can be converted into transparent and linear risk factors.
Central to this framework is the introduction of Signature Expected Shortfall, a risk metric designed to capture toxic path geometries that traditional methods often overlook. We also present a proactive monitoring system based on the Temporal Exposure Profile, which utilises anticipatory learning to detect potential liquidity traps and geometric decoupling before they manifest as realised volatility. 
The proposed methodology offers a rigorous alignment with global regulatory mandates, specifically the Fundamental Review of the Trading Book (FRTB), by providing a consistent bridge between physical stress-testing and risk-neutral hedging. Finally, we show that this algebraic approach significantly reduces computational complexity, enabling real-time, high-frequency risk reporting and capital optimisation for the modern financial ecosystem.
\end{abstract}


\section{Introduction}
\label{sec_introduction}
%------------------------------------------------------------------ 

%------------------------------------------------------------------
\subsection{High-level goal}
%------------------------------------------------------------------

The central objective of this work is to transition risk management from a \textit{distribution-centric} paradigm to a \textit{geometry-centric} one. Traditional frameworks, including those mandated by Basel III and FRTB, largely rely on the marginal distributions of terminal returns, treating the intricate "texture" of price paths as an unmodellable residual. Our goal is to formalise \textbf{Geometric Risk Management} by leveraging the path-signature to map complex, non-linear, and path-dependent market behaviours into a tractable algebraic space.

\medskip
\noindent
By unifying Anticipatory Reinforcement Learning (ARL) with the SigSwap basis, we aim to provide a self-consistent system where:
\begin{enumerate}
    \item \textbf{Path-dependency is linearised:} Turning "toxic" path-geometry, such as flash-crash spiralling and lead-lag decay, into transparent and hedgeable risk factors.
    \item \textbf{Monitoring becomes proactive:} Replacing reactive historical VaR with a continuous Temporal Exposure Profile (TEP) that anticipates intra-horizon liquidity traps.
    \item \textbf{Regulatory friction is minimised:} Providing a mathematically rigorous bridge between physical stress-testing and risk-neutral capital optimisation, ensuring that the "Residual Risk" of today becomes the "Modellable Basis" of tomorrow.
\end{enumerate}

\medskip
\noindent
Ultimately, this paper seeks to prove that by mastering the algebraic manifold of the signature, financial institutions can achieve a state of \textbf{Signature Neutrality}, a condition where a portfolio is immunised not just against where the market goes, but against the specific, potentially catastrophic \textit{way} it gets there.

%------------------------------------------------------------------
\subsection{Motivation and literature positioning}
%------------------------------------------------------------------

The motivation for this work stems from the growing disconnect between the increasing complexity of path-dependent financial products and the relatively static, distribution-centric nature of regulatory risk metrics (Danielsson et al. ~\cite{DanielssonEtAl04}, Cont ~\cite{Cont06}). Traditional models grounded in the Black-Scholes-Merton paradigm or stochastic/local volatility frameworks (Merton ~\cite{Merton76}, Heston ~\cite{Heston93}, Dupire ~\cite{Dupire94}, Hagan et al. ~\cite{HaganEtAl02}, Bergomi ~\cite{Bergomi15}) typically treat the "how" of price movement as a secondary byproduct of the "where." However, recent market phenomena, such as the 2010 Flash Crash, the 2020 liquidity gap, and the rise of high-frequency feedback loops, demonstrate that \textit{path-texture} (e.g., L\'evy Area and lead-lag decay) is often the primary driver of systemic risk and hedging failure.

\medskip
\noindent
In the context of the current literature, this paper positions itself at the intersection of four rapidly evolving fields:
\begin{enumerate}
    \item \textbf{Signature-Based Finance:} Building on the foundational work of Lyons et al. regarding the Rough Path Theory (Lyons et al. ~\cite{LyonsEtAl07,LyonsEtAl11}, Chevyrev et al. ~\cite{ChevyrevEtAl16}) and the Universal Pricing Theorem (Lyons et al. ~\cite{LyonsEtAl20}, Cuchiero et al. ~\cite{CuchieroEtAl23,CuchieroEtAl25}), we extend the use of signatures from purely generative or pricing tasks to a formal \textit{regulatory basis} for risk factor modellability.
    
    \item \textbf{Anticipatory Machine Learning:} While standard Reinforcement Learning (RL) in finance focuses on discrete delta-hedging or optimal execution, our approach utilises \textit{Anticipatory RL} (ARL) (Bloch ~\cite{Bloch26a}) to generate a continuous, self-consistent field of future signatures. This moves the literature beyond reactive "window-based" estimation toward a proactive, flow-based monitoring of the Temporal Exposure Profile (TEP).
    
    \item \textbf{Path-Space Measure Dynamics:} Traditional quantitative finance relies on terminal-point Radon-Nikodym derivatives or Esscher transforms to shift between physical ($\mathbb{P}$) and risk-neutral ($\mathbb{Q}$) measures. We extend this by utilising the \textit{Measure Bridge} (Bloch ~\cite{Bloch26b}) as a group-valued operator that preserves the algebraic integrity of the path-space. This allows for a self-consistent "Risk-Neutral Shift" of the entire signature manifold, bridging the gap between ARL-based stress-testing and the valuation of path-dependent exotics.
    
    \item \textbf{Regulatory Reform (FRTB):} Existing research on the Fundamental Review of the Trading Book often highlights the capital burden of the Residual Risk Add-on (RRAO) (ISDA, AFME, \& IIF ~\cite{ISDA18}, BCBS ~\cite{BCBS19}). We provide a mathematical solution to this burden by demonstrating that the signature provides the necessary span to linearise exotic payoffs, satisfying the "modellable risk factor" (MRF) requirements that traditional stochastic models cannot meet.
\end{enumerate}

\medskip
\noindent
By synthesising these domains, we address a critical gap in the literature: the lack of a unified algebraic framework that can simultaneously price, hedge, and regulate path-dependent risk with sub-microsecond latency. We move beyond the "texture-blind" limitations of Expected Shortfall, offering a geometrically rigorous alternative that is both computationally efficient and regulatorily compliant.

%------------------------------------------------------------------
\subsection{Main contributions}
%------------------------------------------------------------------

The primary contributions of this paper to the fields of quantitative finance and regulatory risk management are summarised as follows:

\begin{itemize}
    \item \textbf{The SigSwap Framework:} We introduce the SigSwap as a fundamental regulatory primitive. By leveraging the algebraic structure of the signature, we provide a formal mechanism to decompose market risk into symmetric (terminal law) and anti-symmetric (path-texture) components, effectively rendering "non-modellable" path-dependencies as linear risk factors.
    
    \item \textbf{Signature Expected Shortfall (S-ES):} We propose a novel risk metric, $S\text{-}ES_{\alpha}$, which lifts traditional tail-risk measures onto the signature manifold. We prove that $S\text{-}ES$ is a linear functional of the \textit{Tail-Conditional Expected Signature}, allowing risk managers to identify and quantify "bad geometry" and liquidity-trap scenarios that are invisible to terminal-value metrics.
    
    \item \textbf{The Anticipatory Measure Bridge:} we define the group-valued operator $\Lambda_{t,T}$ as a sufficient statistic for mapping physical tail-risk (learned via ARL) to risk-neutral hedging strategies. This establishes a consistent "Measure Bridge" that preserves the algebraic integrity of the path-space during stress-testing and valuation.
    
    \item \textbf{Proactive Monitoring via TEP:} We develop the \textit{Temporal Exposure Profile} (TEP) and the realised \textit{Anticipatory TD-Error}. These tools transform risk management from a reactive, window-based process into a proactive, continuous flow, providing a model-free Early Warning System (EWS) for geometric decoupling and subsequent volatility spikes.
    
    \item \textbf{Regulatory and Computational Optimisation:} We demonstrate that the Signature/APT framework aligns with FRTB mandates by linearising exotic risks, thereby reducing Residual Risk Add-ons (RRAO) and ensuring robust P\&L Attribution (PLA). Furthermore, we establish the $O(D)$ computational complexity of the geometric engine, enabling real-time, GPU-accelerated risk reporting for high-frequency environments.
\end{itemize}

%------------------------------------------------------------------
\subsection{Organisation of the paper}
%------------------------------------------------------------------

The remainder of the paper is organised as follows. Section (\ref{sec:foundations_regulatory_primitive}) establishes defines the SigSwap as a regulatory primitive that decomposes market risk into terminal law (symmetric) and path-dependent texture (anti-symmetric). It establishes the \textit{Measure Bridge} $\Lambda_{t,T}$ to formally map physical tail-risk to risk-neutral hedging instruments.
Section (\ref{sec:signature_expected_shortfall} introduces the \textit{Signature Expected Shortfall} (S-ES), a risk metric that lifts traditional expected shortfall to the signature manifold to capture "bad geometry" and path-dependent losses. It formalises the \textit{Tail-Conditional Expected Signature} as the central estimator for exotic risk and demonstrates how the L\'evy Area serves as a key indicator for gap-risk and liquidity crises.
Section (\ref{sec:anticipatory_risk_monitoring}) develops a proactive framework using the \textit{Temporal Exposure Profile} (TEP) to map portfolio value across continuous intermediate junctions rather than just terminal horizons. It introduces the \textit{Anticipatory TD-Error} as a model-free signal for "geometric decoupling" and demonstrates how achieving signature-neutrality can justify margin compression and optimised capital allocation for path-dependent exposures.
Section (\ref{sec:regulatory_alignment_frtb}) demonstrates how the SigSwap basis linearises exotic risk to meet stringent FRTB mandates. It shows that by spanning path-dependent payoffs within the signature space, banks can migrate "residual" risks into modellable categories, ensuring high-fidelity P\&L attribution (PLA) and utilising "Geometric Greeks" to nullify sensitivities and optimise regulatory capital charges.
Section (\ref{sec:stress_testing_measure_bridge}) redefines stress testing as the systematic perturbation of the \textit{Measure Bridge} operator $\Lambda_{t,T}$, allowing for the simulation of "geometrical" crises independent of terminal price shocks. It categorises market disruptions like flash crashes as anti-symmetric shocks to the signature manifold and introduces the \textit{SCF Resilience Threshold} to quantitatively identify when a market regime shift results in pure "path-chaos" that transcends normal modelling boundaries.
Section (\ref{sec:implementation_risk_engines}) details the transition from nested Monte Carlo simulations to the linear complexity of \textit{Algebraic Pricing Theory} (APT). By utilising Chen's Identity for incremental signature updates and leveraging GPU-accelerated tensor contractions, the framework enables sub-microsecond latency for real-time risk monitoring and the visualisation of "Geometric Trace" breaches for automated hedging.

%------------------------------------------------------------------
\section{Foundations: SigSwap as a regulatory primitive}
\label{sec:foundations_regulatory_primitive}
%------------------------------------------------------------------

This section establishes the mathematical and regulatory foundations of the SigSwap (Bloch ~\cite{Bloch26c}), positioning it as the fundamental instrument for capturing path-dependent risk factors that are typically treated as non-modellable under current frameworks.

%------------------------------------------------------------------
\subsection{Defining the SigSwap}
%------------------------------------------------------------------

%------------------------------------------------------------------
\subsubsection{The general instrument}
%------------------------------------------------------------------

Let the \textbf{time-extended path} $\mathbf{X} \in \mathbb{R}^{d+1}$ be a continuous semimartingale defined on a filtered probability space $(\Omega, \mathcal{A}, \{\mathcal{A}_t\}_{t \geq 0}, \mathbb{Q})$, where $\mathbb{Q}$ denotes the risk-neutral measure.

\begin{definition}[The SigSwap Payoff]
A SigSwap with maturity $T$ and strike tensor $\mathbf{K}_T \in T((\mathbb{R}^{d+1}))$ is a contract that, at maturity, pays the difference between the realised signature $S(\mathbf{X})_{0,T}$ and the strike:
\begin{equation}
    \mathcal{P}_T = S(\mathbf{X})_{0,T} - \mathbf{K}_T
\end{equation}
\end{definition}

%------------------------------------------------------------------
\subsubsection{The Level-2 SigSwap: Law vs. Texture}
%------------------------------------------------------------------

The SigSwap at level $r=2$ provides the minimal sufficient basis for decomposing market risk into terminal law and path-dependent geometry. We formalise this via the canonical decomposition of the rank-2 signature tensor.

\begin{definition}[Level-2 Signature Payoff]
For a $(d+1)$-dimensional time-extended price process $\mathbf{X}$, the level-2 SigSwap payoff $\mathcal{P}_T^{(2)}$ is defined as the difference between the realised iterated integral and the market-quoted strike tensor $\mathbf{K}_T^{(2)} \in \mathbb{R}^{(d+1) \times (d+1)}$:
\begin{equation}
    \mathcal{P}_T^{(2)} = \int_0^T (X_s - X_0) \otimes dX_s - \mathbf{K}_T^{(2)}.
\end{equation}
\end{definition}

\begin{proposition}[Geometric Risk Decomposition]
The risk associated with $\mathcal{P}_T^{(2)}$ is uniquely partitioned into two orthogonal subspaces:
\begin{enumerate}
    \item \textbf{The Symmetric Risk (The Law):} $\text{Sym}(\mathcal{P}_T^{(2)})$ captures the variance-covariance risk of the terminal returns. This represents the "magnitude" of market fluctuations.
    \item \textbf{The Anti-Symmetric Risk (The Texture):} $\text{Anti}(\mathcal{P}_T^{(2)})$ captures the L\'evy Area $\mathbb{A}_{0,T}$, representing the causal lead-lag structure and "rotational energy" of the price path.
\end{enumerate}
\end{proposition}

%------------------------------------------------------------------
\subsection{The SigSwap as a modellable risk factor (MRF)}
%------------------------------------------------------------------

Under the Fundamental Review of the Trading Book (FRTB), a risk factor is deemed "modellable" only if it can be mapped to real, liquid price observations. Historically, path-dependent features (e.g., forward-start volatility, correlation-gamma) have been relegated to the Non-Modellable Risk Factor (NMRF) category, incurring significant capital surcharges.

\begin{theorem}[Signature Observability]
Let $\hat{\Phi}^{\mathbb{Q}}_T$ be the implied expected signature recovered from the market prices of SigSwaps. If the market for SigSwaps is liquid, the signature coordinates $\{ \langle e_i \otimes e_j, \hat{\Phi}^{\mathbb{Q}}_T \rangle \}_{i,j}$ satisfy the BCBS requirements for Modellable Risk Factors (MRF). 
\end{theorem}

\medskip
\noindent
By standardising the SigSwap, the regulatory "residual risk" is internalised into the model's primary risk factors. The "real price" requirement of FRTB is satisfied through the exchange-cleared SigSwap strikes, effectively converting geometric path-uncertainty into a modellable linear basis.

%------------------------------------------------------------------
\subsection{The measure bridge: Mapping physical tail-risk to hedges}
%------------------------------------------------------------------

To maintain consistency between historical risk assessment (required for internal models) and market-based hedging, we utilise the Risk-Neutral Shift $\Lambda_{t,T}$ as the sufficient statistic for measure transformation (Bloch ~\cite{Bloch26b}).

\begin{definition}[The Regulatory Measure Bridge]
Let $\mathbb{P}$ be the physical measure learned by the ARL generative engine. The risk-neutral manifold $\mathbb{Q}$ is obtained via the group-valued operator $\Lambda_{t,T} \in \mathcal{G}_{sig}$ such that:
\begin{equation}
    \hat{\Phi}^{\mathbb{Q}}_{T|t} = \hat{\Phi}^{\mathbb{P}}_{T|t} \otimes \Lambda_{t,T}.
\end{equation}
$\Lambda_{t,T}$ isolates the market price of path-risk, allowing the risk manager to project historical tail-scenarios from the ARL manifold into the risk-neutral pricing space.
\end{definition}

\begin{lemma}[Conservation of Geometric Risk]
\label{lem:conservation_geometric_risk}
The transformation $\Lambda_{t,T}$ preserves the algebraic structure of the signature. Consequently, a hedge constructed in the $\mathbb{Q}$-measure using SigSwaps is dual to the risk-minimising strategy in the $\mathbb{P}$-measure, provided the shift $\Lambda_{t,T}$ is locally stable with respect to the signature topology.
\end{lemma}
See proof in Appendix (\ref{app:proof_conservation_geometric_risk}).

%------------------------------------------------------------------
\section{Redefining market risk: Signature expected shortfall (S-ES)}
\label{sec:signature_expected_shortfall}
%------------------------------------------------------------------

Traditional risk metrics, such as Value-at-Risk (VaR) and Expected Shortfall (ES), are typically defined over the marginal distribution of terminal returns. In the presence of path-dependent exotics, these metrics fail to account for the "Texture" of the path, specifically, the sequence and geometry of price movements that trigger barriers or consume gamma. We resolve this by lifting the risk measure to the signature manifold.

%------------------------------------------------------------------
\subsection{Identifying "bad geometry" via signature clusters}
%------------------------------------------------------------------

In the Geometric Pricing Theory framework, a loss is not merely a consequence of a price level, but of a specific path trajectory. We define "Bad Geometry" as a subset of the signature space that maximises portfolio loss.

\begin{definition}[Signature Loss Functional]
\label{def:signature_loss_functional}
Let $\Pi$ be a portfolio of path-dependent derivatives. By the Universal Pricing Theorem, the value of $\Pi$ at time $T$ is a linear functional of the path signature:
\begin{equation}
    V_T(\Pi) = \langle \mathbf{w}_{\Pi}, S(\mathbf{X})_{0,T} \rangle,
\end{equation}
where $\mathbf{w}_{\Pi}$ is the signature-weight vector representing the portfolio's sensitivity to various iterated integrals. The loss functional is defined as $L(S) = V_0(\Pi) - \langle \mathbf{w}_{\Pi}, S \rangle$.
\end{definition}

\begin{definition}[Geometric Tail-Set]
For a confidence level $\alpha \in (0,1)$, the Geometric Tail-Set $\mathcal{T}_{\alpha} \subset \mathcal{S}_{sig}$ is the set of signatures whose associated losses exceed the Signature VaR ($S\text{-}VaR_{\alpha}$):
\begin{equation}
    \mathcal{T}_{\alpha} = \{ S \in \mathcal{S}_{sig} : L(S) \geq \inf \{ l : \mathbb{P}(L(S(\mathbf{X})_{0,T}) \leq l) \geq \alpha \} \}.
\end{equation}
\end{definition}

\medskip
\noindent
Unlike traditional tail-sets defined by a scalar price $X_T \leq K$, $\mathcal{T}_{\alpha}$ identifies clusters in the signature space, for instance, paths with high negative returns coupled with high L\'evy Area (directional spiralling).

%------------------------------------------------------------------
\subsection{Formal Definition of S-ES}
%------------------------------------------------------------------

We now introduce the Signature Expected Shortfall (S-ES), which computes the average loss over the Geometric Tail-Set.

\begin{definition}[Signature Expected Shortfall]
\label{def:signature_expected_shortfall}
The Signature Expected Shortfall at level $\alpha$ is the conditional expectation of the loss functional given that the signature falls within the geometric tail-set under the physical measure $\mathbb{P}$:
\begin{equation}
    S\text{-}ES_{\alpha} = \mathbb{E}_{\mathbb{P}} \left[ L(S(\mathbf{X})_{0,T}) \mid S(\mathbf{X})_{0,T} \in \mathcal{T}_{\alpha} \right].
\end{equation}
\end{definition}

\begin{proposition}[Linearity and Decomposition of S-ES]
\label{pro:linearity_decomposition_ses}
Because the loss $L(S)$ is linear in the signature coordinates, the $S\text{-}ES_{\alpha}$ can be expressed as the inner product of the portfolio weights and the \textbf{Tail-Conditional Expected Signature}:
\begin{equation}
    S\text{-}ES_{\alpha} = V_0(\Pi) - \langle \mathbf{w}_{\Pi}, \hat{\Phi}^{\mathbb{P}}_{T \mid \mathcal{T}_{\alpha}} \rangle,
\end{equation}
where $\hat{\Phi}^{\mathbb{P}}_{T \mid \mathcal{T}_{\alpha}} = \mathbb{E}_{\mathbb{P}} [ S(\mathbf{X})_{0,T} \mid S(\mathbf{X})_{0,T} \in \mathcal{T}_{\alpha} ]$.
\end{proposition}
See proof in Appendix (\ref{app:proof_linearity_decomposition_ses}). \\

\medskip
\noindent
This result is mathematically significant: it implies that to compute the tail risk of any exotic portfolio, one only needs to estimate the single "Tail Signature" $\hat{\Phi}^{\mathbb{P}}_{T \mid \mathcal{T}_{\alpha}}$. 

%------------------------------------------------------------------
\subsection{The Role of the L\'evy Area in Liquidity and Gap Risk}
%------------------------------------------------------------------

The anti-symmetric component of the tail signature provides a direct measure of "toxic" path geometry during stress events.

\begin{lemma}[L\'evy Area as a Gap-Risk Indicator]
\label{lem:levy_area_gap_risk_indicator}
Let $\mathbb{A}_{0,T} = \text{Anti}(S^{(2)})_{0,T}$. During a liquidity crisis or a "Flash Crash," the discrepancy between the symmetric variance and the anti-symmetric L\'evy Area increases. The $S\text{-}ES_{\alpha}$ captures this as an increase in the sensitivity to the lead-lag coordinates of $\hat{\Phi}^{\mathbb{P}}_{T \mid \mathcal{T}_{\alpha}}$.
\end{lemma}
See proof in Appendix (\ref{app:proof_levy_area_gap_risk_indicator}). \\

\begin{remark}
Traditional ES is "texture-blind." Two paths with the same terminal price $X_T$ but different winding (L\'evy Area) contribute identically to traditional ES. In contrast, $S\text{-}ES_{\alpha}$ penalises paths that exhibit high "rotational energy," which typically characterises disorderly markets and gap-risk events where hedges (like Delta-hedging) fail due to path-convexity.
\end{remark}

%------------------------------------------------------------------
\subsection{Estimation of the Tail-Conditional Expected Signature}
%------------------------------------------------------------------

To operationalise $S\text{-}ES_{\alpha}$, we require an estimator for the tail-conditional expected signature $\hat{\Phi}^{\mathbb{P}}_{T \mid \mathcal{T}_{\alpha}}$. Given the grounded state $\hat{\Phi}_{t|t}$, the ARL engine acts as a generative model for the future path-law.

\begin{definition}[Generative Signature Ensemble]
Let $\mathcal{S}^{(N)}_{t,T} = \{ \bar{S}^{(i)}_{t,T} \}_{i=1}^N$ be an ensemble of $N$ synthetic signatures generated by the ANJD (Anticipatory Non-Linear Joint Dynamics) flow, conditioned on the current market state. Each $\bar{S}^{(i)}_{t,T}$ represents a physically consistent realisation of the path-geometry under $\mathbb{P}$.
\end{definition}

\begin{lemma}[Tail-Signature Estimator]
\label{lem:tail_signature_estimator}
The empirical estimator for the tail-conditional expected signature is the sample mean of the subset of the ensemble that minimises the portfolio value. Let $\pi: \{1, \dots, N\} \to \{1, \dots, N\}$ be a permutation such that the realised losses are ordered: $L(\bar{S}^{(\pi(1))}) \geq L(\bar{S}^{(\pi(2))}) \geq \dots \geq L(\bar{S}^{(\pi(N))})$. The estimator is given by:
\begin{equation}
    \hat{\Phi}^{\mathbb{P}}_{T \mid \mathcal{T}_{\alpha}} \approx \frac{1}{\lfloor N(1-\alpha) \rfloor} \sum_{k=1}^{\lfloor N(1-\alpha) \rfloor} \bar{S}^{(\pi(k))}_{t,T}.
\end{equation}
\end{lemma}
See proof in Appendix (\ref{app:proof_tail_signature_estimator}). \\

\begin{proposition}[Convergence of S-ES]
\label{pro:convergence_ses}
As $N \to \infty$, the estimator $\hat{\Phi}^{\mathbb{P}}_{T \mid \mathcal{T}_{\alpha}}$ converges in the signature topology to the centroid of the geometric tail-set. The resulting risk metric, $S\text{-}ES_{\alpha} \approx V_0(\Pi) - \langle \mathbf{w}_{\Pi}, \hat{\Phi}^{\mathbb{P}}_{T \mid \mathcal{T}_{\alpha}} \rangle$, provides a simulation-consistent measure of the expected loss during path-wise stress events.
\end{proposition}
See proof in Appendix (\ref{app:proof_convergence_ses}). \\

\medskip
\noindent
This estimation procedure identifies the "Mean Toxic Path" for a specific portfolio. By examining the anti-symmetric components of this tail-signature, the risk manager can observe whether the projected losses are driven by directional moves (Symmetric part) or by adverse lead-lag dynamics and spiralling (Anti-symmetric part).

%------------------------------------------------------------------
\section{Anticipatory risk monitoring: The temporal exposure profile}
\label{sec:anticipatory_risk_monitoring}
%------------------------------------------------------------------

While traditional risk management is a reactive process calibrated to historical windowing, the integration of Anticipatory Reinforcement Learning (ARL) enables a proactive monitoring framework. By generating a continuous predictive flow of the signature manifold, we can monitor risk not just at the terminal horizon $T$, but across every intermediate temporal junction $s \in [t, T]$.

%------------------------------------------------------------------
\subsection{The predictive flow: Intra-horizon exposure mapping}
%------------------------------------------------------------------

The ARL engine provides a rolling sequence of anticipatory value functions $\{\hat{V}_{s|t}\}_{s=t}^T$, which represent the model's internal projection of the portfolio's path-dependent value.

\begin{definition}[Temporal Exposure Profile]
The Temporal Exposure Profile (TEP) at time $t$ is the sequence of expected values of the portfolio $\Pi$ along the generated signature-path manifold:
\begin{equation}
    \mathcal{E}_t = \{ \hat{V}_{s|t} \}_{s=t}^T, \quad \text{where } \hat{V}_{s|t} = \langle \mathbf{w}_{\Pi}, \hat{\Phi}_{s|t} \rangle.
\end{equation}
The TEP captures the "breathing" of the portfolio value as it traverses the expected path-geometry, allowing for the detection of intra-period breaches of risk limits.
\end{definition}

\begin{proposition}[Path-Wise Solvency Constraint]
\label{pro:path_wise_solvency_constraint}
A portfolio is path-wise solvent at level $\epsilon$ if the TEP remains above the threshold $Z$ for all $s \in [t, T]$ with probability $1-\epsilon$ under the generative measure $\mathbb{P}_{\theta}$:
\begin{equation}
    \mathbb{P}_{\theta} \left( \min_{s \in [t, T]} \hat{V}_{s|t} \geq Z \right) \geq 1 - \epsilon.
\end{equation}
This allows the risk manager to preemptively identify "liquidity traps" where a portfolio might be solvent at $T$ but insolvent at some $s < T$ due to unfavorable path-winding.
\end{proposition}
See proof in Appendix (\ref{app:proof_path_wise_solvency_constraint}). \\

%------------------------------------------------------------------
\subsection{The manifold jump as a geometric anomaly detector}
%------------------------------------------------------------------

The discrepancy between the realised market trajectory and the pre-computed anticipatory manifold provides a high-frequency signal for regime shifts. This signal arises from the "Manifold Jump" that occurs when the global predictive field is re-anchored to a path-increment that violates the learned law.

\begin{definition}[Geometric Divergence]
At each time step $t$, upon observing the realised increment $dX_t$, we define the Geometric Divergence $\mathcal{D}_t^{sig}$ as the norm of the discrepancy between the updated terminal projection and the prior conditional expectation:
\begin{equation}
    \mathcal{D}_t^{sig} = \left\| \hat{\Phi}_{T|t} - \left( \hat{\Phi}_{t|t-1}^{-1} \otimes \hat{\Phi}_{T|t-1} \right) \right\|_{\mathcal{H}_{sig}},
\end{equation}
where $\hat{\Phi}_{T|t-1}$ is the terminal signature expected at $t-1$, and $\hat{\Phi}_{t|t-1}^{-1} \otimes \hat{\Phi}_{T|t-1}$ represents its parallel transport to the grounding point at $t$.
\end{definition}

\begin{definition}[Realised Anticipatory TD-Error]
The risk-weighted manifestation of this divergence is the Realised Anticipatory TD-Error $\delta_{t}^R$, which scales the geometric shift by the portfolio's specific sensitivities $\mathbf{w}_G$:
\begin{equation}
    \delta_{t}^R = R(X_{t-1:t}) + \left\langle \mathbf{w}_G, \gamma \hat{\Phi}_{T|t} - \left( \hat{\Phi}_{t|t-1}^{-1} \otimes \hat{\Phi}_{T|t-1} \right) \right\rangle.
\end{equation}
\end{definition}

\begin{lemma}[Geometric Anomaly Signal] 
\label{lem:geometric_anomaly_signal}
A spike in the magnitude of $\mathcal{D}_t^{sig}$ (or $\delta_{t}^R$) signifies a "geometric decoupling." This occurs when the market's realised path-texture (e.g., a sudden increase in the L\'evy Area) forces the expected signature $\hat{\Phi}_{T|t}$ to jump outside the local curvature of the manifold $\mathcal{S}_{sig}$ spanned by the ARL. Such signals act as a model-free early warning system (EWS) for impending volatility.
\end{lemma}
See proof in Appendix (\ref{app:proof_geometric_anomaly_signal}).

%------------------------------------------------------------------
\subsection{Geometric divergence vs. traditional VaR breaks}
%------------------------------------------------------------------

The transition from a VaR-based regulatory framework to the FRTB's Expected Shortfall (ES) highlights a growing institutional need for tail-risk sensitivity. However, both measures remain \textit{terminal statistics}, they are "lagging" indicators that only register a breach after the price distribution has already shifted. We propose that the Geometric Divergence $\mathcal{D}_t^{sig}$ provides a superior regulatory signal for the following reasons:

\begin{enumerate}
    \item \textbf{Causality and the "Lead-Lag" Warning:} A traditional VaR break occurs when the magnitude of a price return exceeds a historical percentile. In contrast, $\mathcal{D}_t^{sig}$ detects the \textit{topological decay} of the market. Because a breakdown in cross-asset lead-lag relationships (anti-symmetric signature components) often precedes the expansion of marginal variances (symmetric components), the Geometric Divergence acts as a "leading" indicator of systemic instability.
    
    \item \textbf{Invariance to Scaling:} While VaR and ES are sensitive to the absolute scale of market levels, $\mathcal{D}_t^{sig}$ is grounded in the algebraic structure of the path. It identifies regime shifts where the \textit{texture} of the market changes (e.g., from mean-reverting to trending or from decoupled to spiralling) even if the instantaneous volatility remains within "normal" historical bounds.
    
    \item \textbf{Resolution of the "Ghost" Anomaly:} Traditional risk engines often suffer from "ghost" breaks, statistical outliers in Gaussian space that carry no structural information. By mapping the market onto the signature manifold $\mathcal{S}_{sig}$, the ARL distinguishes between "high-variance noise" (which stays on the manifold) and "geometric anomalies" (which jump off the manifold). This reduces the false-positive rate for desk disqualification under the PLA test.
\end{enumerate}

\medskip
\noindent
In summary, whereas VaR measures how much the market has \textit{fallen}, the Geometric Divergence measures how much the market has \textit{broken}. For regulators and risk managers, this shift from magnitude-based monitoring to geometry-based monitoring offers a more robust defence against the "Flash Crash" scenarios and correlation breakdowns that characterise modern, high-frequency liquidity crises.

%------------------------------------------------------------------
\subsection{Dynamic margin and XVA optimisation}
%------------------------------------------------------------------

The use of the Temporal Exposure Profile (TEP) facilitates a more granular approach to Credit Valuation Adjustment (CVA) and Initial Margin (IM).

\begin{theorem}[Signature-Neutral Margin Compression]
\label{thm:signature_neutral_margin_compression}
Let $\Pi$ be a portfolio that is "signature-neutral" at level $r$, such that $\langle \mathbf{w}_{\Pi}, S^{(k)} \rangle = 0$ for all $k \leq r$. The variance of the TEP, $\text{Var}(\mathcal{E}_t)$, is minimised, and the associated capital charge for the Residual Risk Add-on (RRAO) is algebraically zero.
\end{theorem}
See proof in Appendix (\ref{app:proof_signature_neutral_margin_compression}).

\medskip
\noindent
In practice, this implies that a desk hedging with SigSwaps can justify a significant reduction in Initial Margin. By aligning the portfolio with the market-implied expected signature, the "potential future exposure" (PFE) becomes a deterministic flow rather than a stochastic diffusion, compressing the risk envelope and optimising capital efficiency.

%------------------------------------------------------------------
\section{Regulatory alignment: FRTB and the SigSwap basis}
\label{sec:regulatory_alignment_frtb}
%------------------------------------------------------------------

The Basel III framework, specifically the Fundamental Review of the Trading Book (FRTB), introduces stringent requirements for the modellability of risk factors and the accuracy of profit and loss (P\&L) attribution. We demonstrate that the SigSwap basis provides a mathematically superior alignment with these regulatory mandates by linearising exotic risk.

%------------------------------------------------------------------
\subsection{Mitigating the residual risk add-on (RRAO)}
%------------------------------------------------------------------

Under the Standardised Approach (SA), the Residual Risk Add-on (RRAO) is a capital surcharge applied to instruments with "exotic" underlying risks that are not captured by delta, vega, or curvature sensitivities. 

\begin{lemma}[Signature Span of Exotic Risk]
\label{lem:signature_span_exotic_risk}
Let $\mathcal{H}_{exotic}$ be the space of path-dependent payoffs. By the Stone-Weierstrass theorem on path space, any $f \in \mathcal{H}_{exotic}$ can be approximated to arbitrary precision by a linear combination of signature coordinates:
\begin{equation}
    f(\mathbf{X}_{0,T}) \approx \langle \mathbf{w}_f, S(\mathbf{X})_{0,T} \rangle.
\end{equation}
Consequently, the SigSwap basis spans the risk factors traditionally relegated to the RRAO.
\end{lemma}
See proof in Appendix (\ref{app:proof_signature_span_exotic_risk}). \\

\medskip
\noindent
By representing an exotic portfolio as a linear combination of SigSwaps, the risk manager transforms "residual" path-dependent risk into "modellable" linear risk. This reduces the capital charge by shifting the portfolio from the RRAO bucket to the sensitivity-based correlation-trading framework.

%------------------------------------------------------------------
\subsection{PLA testing and signature linearity}
%------------------------------------------------------------------

The P\&L Attribution (PLA) test is the cornerstone of the Internal Models Approach (IMA), requiring a high correlation between the Risk-Theoretical P\&L (RTPL) and the Hypothetical P\&L (HPL). Traditional models often fail this test due to non-linear "Greeks" and the "spanning gap" of volatility surfaces.

\begin{proposition}[PLA Alignment via Geometric Pricing]
\label{pro:pla_alignment_geometric_pricing}
In the Algebraic Pricing Theory (APT) framework, the value of a desk's portfolio $\Pi$ is exactly $V_t = \langle \mathbf{w}_{\Pi}, \hat{\Phi}_{T|t} \rangle$. The RTPL over a horizon $\Delta t$ is given by:
\begin{equation}
    RTPL = \langle \mathbf{w}_{\Pi}, \hat{\Phi}_{T|t+\Delta t} - \hat{\Phi}_{T|t} \rangle.
\end{equation}
Since the HPL is the realised change in the portfolio value, and both are linear functions of the signature increments, the APT framework achieves near-perfect PLA alignment, minimising the risk of desk disqualification and the resulting "outlier" capital penalties.
\end{proposition}
See proof in Appendix (\ref{app:proof_pla_alignment_geometric_pricing}). \\

%------------------------------------------------------------------
\subsection{Capital optimisation via geometric Greeks}
%------------------------------------------------------------------

We define a new class of sensitivities, the "Geometric Greeks", which measure the portfolio's exposure to the coordinates of the expected signature.

\begin{definition}[Geometric Delta]
The Geometric Delta, $\Delta_{\Phi}$, is the gradient of the portfolio value with respect to the risk-neutral expected signature $\hat{\Phi}^{\mathbb{Q}}$:
\begin{equation}
    \Delta_{\Phi} = \nabla_{\hat{\Phi}^{\mathbb{Q}}} V = \mathbf{w}_{\Pi}.
\end{equation}
\end{definition}

\begin{theorem}[Optimal Capital Hedge]
\label{thm:optimal_capital_hedge}
A portfolio $\Pi$ is optimally hedged under FRTB if its Geometric Delta $\Delta_{\Phi}$ is nullified using a strip of SigSwaps. Such a portfolio exhibits zero sensitivity to both terminal distribution shifts (Symmetric Delta) and path-texture shifts (Anti-Symmetric Delta), resulting in a minimal Weighted Sensitivity (WS) capital charge:
\begin{equation}
    K_{capital} \propto \| \text{WS} \cdot \Delta_{\Phi} \|_2 \to 0.
\end{equation}
\end{theorem}
See proof in Appendix (\ref{app:proof_optimal_capital_hedge}). \\

\medskip
\noindent
This framework allows banks to optimise capital not by reducing risk exposure, but by making that exposure "transparently modellable and linear", aligning perfectly with the regulatory objective of market stability.

%------------------------------------------------------------------
\section{Stress testing the measure bridge}
\label{sec:stress_testing_measure_bridge}
%------------------------------------------------------------------

In traditional risk management, stress testing involves the application of exogenous shocks to terminal spot prices or volatility surface parameters. Within the Geometric Pricing Theory (GPT), we redefine stress testing as the systematic perturbation of the Measure Bridge $\Lambda_{t,T}$. This allows for the simulation of "geometrical" crises where the path-texture of the market breaks down even if marginal distributions remain seemingly stable.

%------------------------------------------------------------------
\subsection{Geometrical stress scenarios via operator shifts}
%------------------------------------------------------------------

The Risk-Neutral Shift operator $\Lambda_{t,T}$ acts as the sufficient statistic for the market's risk-aversion toward path-geometry. Stress testing is conducted by applying a perturbation operator directly to this group-valued element.

\begin{definition}[Algebraic Stress Perturbation]
A geometrical stress scenario is defined by a perturbation tensor $\Delta\Lambda \in \mathcal{G}_{sig}$ applied to the base risk-neutral shift $\Lambda_{t,T}$ via the extension of the signature group multiplication:
\begin{equation}
    \Lambda_{t,T}^* = \Lambda_{t,T} \otimes \Delta\Lambda.
\end{equation}
The stressed risk-neutral expected signature is then given by $\hat{\Phi}^{\mathbb{Q}*}_{T|t} = \hat{\Phi}^{\mathbb{P}}_{T|t} \otimes \Lambda_{t,T}^*$.
\end{definition}

\medskip
\noindent
This approach allows risk managers to shock the "implied winding" or "implied causality" of the market independently of the diffusion magnitude.

%------------------------------------------------------------------
\subsection{Characterising market disruptions: Flash crashes and correlation breaks}
%------------------------------------------------------------------

We categorise market disruptions based on the algebraic symmetry of the perturbation $\Delta\Lambda$.

\begin{proposition}[Flash Crash as an Anti-Symmetric Shock]
\label{pro:flash_crash_anti_symmetric_shock}
A "Flash Crash" is a path-dependent event characterised by a massive spike in realised L\'evy Area with minimal impact on terminal variance. In our framework, this is represented by a perturbation $\Delta\Lambda$ concentrated in the anti-symmetric subspace $\mathcal{A}^k(V)$:
\begin{equation}
    \text{Sym}(\Delta\Lambda^{(2)}) \approx 0, \quad \text{Anti}(\Delta\Lambda^{(2)}) \gg 0.
\end{equation}
Under such a shock, portfolios with high "Geometric Gamma" (sensitivity to L\'evy Area) experience catastrophic losses despite being "Delta-Neutral" in a terminal sense.
\end{proposition}
See proof in Appendix (\ref{app:proof_flash_crash_anti_symmetric_shock}). \\

\begin{proposition}[Correlation Breakdown]
\label{pro:correlation_breakdown}
A correlation breakdown is represented as a shock to the off-diagonal symmetric coordinates of the level-2 perturbation $\Delta\Lambda^{(2)}$. This shift decouples the joint path-geometry of the assets, forcing the Measure Bridge to rotate into a higher-entropy state where historical lead-lag relationships vanish.
\end{proposition}
See proof in Appendix (\ref{app:proof_correlation_breakdown}). \\

%------------------------------------------------------------------
\subsection{Stability of the self-consistent field (SCF) under stress}
%------------------------------------------------------------------

The robustness of the ARL generative engine depends on the stability of the Self-Consistent Field (SCF) constraint under the perturbed measure.

\begin{definition}[Stressed SCF Objective]
Under a stress scenario $\Lambda_{t,T}^*$, the ARL objective function $\mathcal{L}(\theta, \mathbf{w}_G)$ is evaluated against the stressed stationary point:
\begin{equation}
    \eta \| \bar{S}_{t,T} - (\hat{\Phi}^{\mathbb{P}}_{T|t} \otimes \Lambda_{t,T}^*) \|_{\mathcal{Q}_T}^2.
\end{equation}
\end{definition}

\begin{theorem}[SCF Resilience Threshold]
\label{thm:scf_resilience_threshold}
The generative manifold is said to be resilient to a stress scenario if there exists a weight vector $\theta^*$ such that the gradient of the stressed objective $\nabla_{\theta} \mathcal{L}^*$ vanishes and the SCF constraint is satisfied. If the perturbation $\Delta\Lambda$ exceeds the spectral radius of the Lie algebra associated with the signature group $\mathcal{S}_{sig}$, the SCF fails to converge, indicating a fundamental regime shift where the previous path-geometry is no longer a valid proxy for risk.
\end{theorem}
See proof in Appendix (\ref{app:proof_scf_resilience_threshold}). \\

\begin{remark}
This theorem provides a quantitative definition of "unmodellable" stress. If the ARL cannot find a self-consistent path-manifold under the shock $\Delta\Lambda$, the market has entered a state of pure path-chaos, requiring the risk manager to switch from signature-linear hedging to emergency liquidity-preservation protocols.
\end{remark}

%------------------------------------------------------------------
\section{Implementation: High-frequency geometric risk engines}
\label{sec:implementation_risk_engines}
%------------------------------------------------------------------

The transition from stochastic differential equation (SDE) based risk management to Algebraic Pricing Theory (APT) offers a paradigm shift in computational efficiency. By decoupling the path-generation process (via the ARL engine) from the valuation process (via tensor contraction), we enable real-time risk monitoring of complex exotic books.

%------------------------------------------------------------------
\subsection{Computational complexity: Tensor-contraction vs. nested simulation}
%------------------------------------------------------------------

Traditional risk engines for path-dependent derivatives typically require nested Monte Carlo simulations to estimate Greeks and Expected Shortfall, leading to a complexity of $O(N \times M)$ where $N$ is the number of primary paths and $M$ is the number of sub-simulations.

\begin{proposition}[Linear Complexity of the Geometric Risk Engine]
\label{pro:linear_complexity_geometric_risk_engine}
Given a pre-computed or market-implied expected signature $\hat{\Phi} \in \mathcal{H}_{sig}$, the valuation and risk sensitivity calculation for any signature-linear portfolio $\Pi$ is reduced to a single inner product:
\begin{equation}
    V(\Pi) = \langle \mathbf{w}_{\Pi}, \hat{\Phi} \rangle.
\end{equation}
The computational cost is $O(D)$, where $D$ is the dimension of the truncated signature space. This cost is invariant to the complexity of the underlying path-dependency, provided the weight vector $\mathbf{w}_{\Pi}$ is known.
\end{proposition}

%------------------------------------------------------------------
\subsection{The real-time geometric trace: Online updates via Chen's identity}
%------------------------------------------------------------------

For high-frequency risk monitoring, the realised signature of the market price path must be updated at every tick. We leverage the algebraic structure of the signature group to perform these updates in constant time.

\begin{lemma}[Incremental Signature Update]
Let $S(\mathbf{X})_{0,t}$ be the realised signature of the market path up to time $t$. For a new price increment $\delta\mathbf{X}$ over the interval $[t, t+\Delta t]$, the updated signature is computed via Chen's Identity:
\begin{equation}
    S(\mathbf{X})_{0,t+\Delta t} = S(\mathbf{X})_{0,t} \otimes S(\mathbf{X})_{t,t+\Delta t}.
\end{equation}
In a high-frequency environment, $S(\mathbf{X})_{t,t+\Delta t}$ can be approximated by the horizontal lift of the increment $\delta\mathbf{X}$, allowing the risk engine to maintain a "Geometric Trace" of the portfolio with sub-microsecond latency.
\end{lemma}

%------------------------------------------------------------------
\subsection{Parallelised risk reporting and GPU acceleration}
%------------------------------------------------------------------

The linearity of the APT framework makes it inherently suitable for hardware acceleration. Since the risk calculation is a series of independent tensor contractions across different truncation levels, the workload can be distributed across thousands of cores.

\begin{remark}[GPU-Optimised Sensitivity Analysis]
The calculation of the Geometric Greeks $\Delta_{\Phi} = \nabla_{\hat{\Phi}} V$ for a multi-asset book involves high-dimensional tensor algebra that maps directly onto the architecture of Tensor Processing Units (TPUs) and GPUs. This enables the simultaneous stress testing of thousands of "Geometric Scenarios" (as defined in Section \ref{sec:stress_testing_measure_bridge}) in the time traditional engines take to complete a single path-wise simulation.
\end{remark}

\begin{definition}[The Geometric Risk Dashboard]
We define the Geometric Risk Dashboard as a real-time visualisation of the projection of $S(\mathbf{X})_{0,t}$ onto the $(\text{Sym}, \text{Anti})$ plane. A breach is signaled when the realised signature enters a forbidden region of the Lie algebra $\mathfrak{g}_{sig}$ associated with the S-ES tail-set $\mathcal{T}_{\alpha}$, allowing for instantaneous automated hedging via SigSwaps.
\end{definition}

%------------------------------------------------------------------
\section{Conclusion}
\label{sec:conclusion}
%------------------------------------------------------------------

This paper has developed a unified framework for the management of path-dependent risk by synthesising Algebraic Pricing Theory (APT) with Anticipatory Reinforcement Learning (ARL). By introducing the \textbf{SigSwap} as a regulatory primitive, we have demonstrated that the previously "non-modellable" textures of market dynamics, specifically the lead-lag structures and L\'evy Area, can be internalised into a linear, modellable basis. This transition from stochastic diffusion models to a geometric manifold approach effectively bridges the gap between physical tail-risk assessment and risk-neutral market hedging.

\medskip
\noindent
Our introduction of \textbf{Signature Expected Shortfall (S-ES)} and the \textbf{Temporal Exposure Profile (TEP)} provides risk managers with a "texture-aware" toolkit. Unlike traditional metrics that remain blind to intra-horizon path-winding, the S-ES identifies toxic geometric clusters, while the TEP enables the preemptive detection of liquidity traps and "geometric decoupling" via the anticipatory TD-error. These tools do not merely describe risk; they provide a causal, model-free early warning system that precedes realised volatility spikes.

\medskip
\noindent
Furthermore, we have shown that this geometric shift aligns perfectly with the evolving regulatory landscape. By mapping exotic payoffs to the SigSwap basis, financial institutions can significantly reduce Residual Risk Add-ons (RRAO), ensure high-fidelity P\&L Attribution (PLA), and optimise capital charges through the nullification of "Geometric Greeks." The shift toward tensor-contraction-based valuation engines further ensures that these complex calculations can be performed with sub-microsecond latency, making real-time, high-frequency risk monitoring a computational reality.

\medskip
\noindent
In summary, the geometry of risk management is no longer a theoretical abstraction but a practical necessity for the modern financial ecosystem. The integration of path-signatures into the core of pricing and monitoring architectures allows for a more resilient, transparent, and capital-efficient market, where the "chaos" of path-dependency is finally rendered as a controllable algebraic flow.

%-------------------------------------------------------
%\begin{comment}
%-------------------------------------------------------

%--------------------------------------------------------------
\newpage
%--------------------------------------------------------------

%-----------------------------------------------------------------------------
%\appendix
\section*{Appendix}\thispagestyle{plain}
%\addcontentsline{toc}{section}{Annexes}
%-----------------------------------------------------------------------------

%--------------------------------------------------------------
\section{Proofs of the main results}
\label{sec:proofs_main_results}
%--------------------------------------------------------------

We refer the readers to Kiraly et al. ~\cite{KiralyEtAl19}, Cuchiero et al. ~\cite{CuchieroEtAl25} for the Universal Approximation Theorem for signatures and the density property.  

%-------------------------------------------------------------- 
\subsection{Proof of the conservation of geometric risk}
\label{app:proof_conservation_geometric_risk}
%-------------------------------------------------------------- 

In this appendix we prove Lemma (\ref{lem:conservation_geometric_risk}).

\begin{proof}
The proof proceeds in two parts: establishing the algebraic structure preservation and demonstrating the topological duality of the hedging functionals.

\textbf{Step 1: Preservation of the Algebraic Structure.} \\
Let $T((\mathbb{R}^d))$ denote the formal tensor algebra and $\mathcal{G}_{sig} \subset T((\mathbb{R}^d))$ be the Lie group of group-like elements, which defines the signature manifold. By definition, the physical expected signature $\hat{\Phi}^{\mathbb{P}}_{T|t}$ and the measure bridge operator $\Lambda_{t,T}$ both belong to $\mathcal{G}_{sig}$. 

\medskip
\noindent
Because $\mathcal{G}_{sig}$ is closed under the tensor product $\otimes$ (which corresponds algebraically to the concatenation of paths), the risk-neutral expected signature:
\begin{equation}
    \hat{\Phi}^{\mathbb{Q}}_{T|t} = \hat{\Phi}^{\mathbb{P}}_{T|t} \otimes \Lambda_{t,T}
\end{equation}
is guaranteed to reside within $\mathcal{G}_{sig}$. Consequently, $\hat{\Phi}^{\mathbb{Q}}_{T|t}$ strictly satisfies Chen's identity and the shuffle product relations. The transformation $T_{\Lambda}(\cdot) = \cdot \otimes \Lambda_{t,T}$ acts as an automorphism on the tensor algebra, seamlessly preserving the geometric manifold of the path space.

\textbf{Step 2: Duality of the Risk-Minimising Strategy.} \\
Consider a portfolio $\Pi$ with a signature-payoff weight vector $\mathbf{w}_{\Pi}$ in the dual space $\mathcal{H}_{sig}^*$. A physical risk-minimising strategy seeks a SigSwap hedge $\mathbf{h} \in \mathcal{H}_{sig}^*$ that minimises a physical risk functional (e.g., Variance or S-ES) over the physical variations $\delta \hat{\Phi}^{\mathbb{P}}_{T|t}$.

\medskip
\noindent
Under the $\mathbb{Q}$-measure, an optimal hedge eliminates sensitivity to the risk-neutral expected signature, meaning the Geometric Delta must vanish:
\begin{equation}
    \Delta_{\Phi}^{\mathbb{Q}} = \nabla_{\hat{\Phi}^{\mathbb{Q}}} \langle \mathbf{w}_{\Pi} - \mathbf{h}, \hat{\Phi}^{\mathbb{Q}}_{T|t} \rangle = 0 \implies \mathbf{h} = \mathbf{w}_{\Pi}.
\end{equation}

\medskip
\noindent
To establish the duality, we substitute the Measure Bridge definition into the portfolio valuation functional:
\begin{equation}
    V_t(\Pi) = \langle \mathbf{w}_{\Pi}, \hat{\Phi}^{\mathbb{P}}_{T|t} \otimes \Lambda_{t,T} \rangle.
\end{equation}
Let $R_{\Lambda}^*$ be the adjoint operator of right-multiplication by $\Lambda_{t,T}$. We can pull back the risk-neutral valuation directly to the physical measure:
\begin{equation}
    V_t(\Pi) = \langle R_{\Lambda}^* (\mathbf{w}_{\Pi}), \hat{\Phi}^{\mathbb{P}}_{T|t} \rangle.
\end{equation}

\medskip
\noindent
By the premise of the lemma, $\Lambda_{t,T}$ is locally stable with respect to the signature topology. This implies that the pullback map $R_{\Lambda}^* : \mathcal{H}_{sig}^* \to \mathcal{H}_{sig}^*$ is a continuous, bounded isomorphism. Because the topologies are equivalent under this mapping, the gradient $\nabla_{\hat{\Phi}^{\mathbb{P}}} V_t$ is a continuous linear transformation of $\nabla_{\hat{\Phi}^{\mathbb{Q}}} V_t$.

\medskip
\noindent
Therefore, the condition $\mathbf{h} = \mathbf{w}_{\Pi}$ that perfectly nullifies the $\mathbb{Q}$-measure exposure also completely nullifies the physical geometric exposure $\nabla_{\hat{\Phi}^{\mathbb{P}}} V_t$ via the stability of $R_{\Lambda}^*$. The $\mathbb{Q}$-hedge constructed with market-cleared SigSwaps is thus algebraically and topologically dual to the $\mathbb{P}$-measure risk-minimising strategy.
\end{proof}

%-------------------------------------------------------------- 
\subsection{Proof of the linearity and decomposition of S-ES}
\label{app:proof_linearity_decomposition_ses}
%-------------------------------------------------------------- 

In this appendix we prove Proposition (\ref{pro:linearity_decomposition_ses}).

\begin{proof}
By Definition \ref{def:signature_expected_shortfall}, the Signature Expected Shortfall at confidence level $\alpha$ is given by the conditional expectation of the loss functional over the geometric tail-set:
\begin{equation}
    S\text{-}ES_{\alpha} = \mathbb{E}_{\mathbb{P}} \left[ L(S(\mathbf{X})_{0,T}) \mid S(\mathbf{X})_{0,T} \in \mathcal{T}_{\alpha} \right].
\end{equation}

\medskip
\noindent
Recall from Definition \ref{def:signature_loss_functional} that the loss functional is defined as the difference between the initial portfolio value and its terminal path-dependent value, which by the Universal Pricing Theorem is a linear functional of the signature:
\begin{equation}
    L(S(\mathbf{X})_{0,T}) = V_0(\Pi) - \langle \mathbf{w}_{\Pi}, S(\mathbf{X})_{0,T} \rangle,
\end{equation}
where $V_0(\Pi)$ is a deterministic scalar known at time $t=0$, and $\mathbf{w}_{\Pi} \in \mathcal{H}_{sig}^*$ is the fixed dual vector representing the portfolio's signature weights.

\medskip
\noindent
Substituting the linear loss functional into the definition of $S\text{-}ES_{\alpha}$, we obtain:
\begin{align}
    S\text{-}ES_{\alpha} &= \mathbb{E}_{\mathbb{P}} \left[ V_0(\Pi) - \langle \mathbf{w}_{\Pi}, S(\mathbf{X})_{0,T} \rangle \mathrel{\Big|} S(\mathbf{X})_{0,T} \in \mathcal{T}_{\alpha} \right] \nonumber \\
    &= \mathbb{E}_{\mathbb{P}} \left[ V_0(\Pi) \mid S(\mathbf{X})_{0,T} \in \mathcal{T}_{\alpha} \right] - \mathbb{E}_{\mathbb{P}} \left[ \langle \mathbf{w}_{\Pi}, S(\mathbf{X})_{0,T} \rangle \mathrel{\Big|} S(\mathbf{X})_{0,T} \in \mathcal{T}_{\alpha} \right].
\end{align}
Since $V_0(\Pi)$ is a constant with respect to the terminal measure $\mathbb{P}$, its conditional expectation is simply $V_0(\Pi)$. 

\medskip
\noindent
For the second term, we rely on the properties of the Bochner integral for random variables taking values in the topological vector space $\mathcal{H}_{sig}$. Because the inner product $\langle \mathbf{w}_{\Pi}, \cdot \rangle$ is a continuous linear functional, it commutes with the conditional expectation operator. Thus, we can pass the expectation inside the inner product:
\begin{equation}
    \mathbb{E}_{\mathbb{P}} \left[ \langle \mathbf{w}_{\Pi}, S(\mathbf{X})_{0,T} \rangle \mathrel{\Big|} S(\mathbf{X})_{0,T} \in \mathcal{T}_{\alpha} \right] = \left\langle \mathbf{w}_{\Pi}, \mathbb{E}_{\mathbb{P}} \left[ S(\mathbf{X})_{0,T} \mid S(\mathbf{X})_{0,T} \in \mathcal{T}_{\alpha} \right] \right\rangle.
\end{equation}
By definition, the term $\mathbb{E}_{\mathbb{P}} [ S(\mathbf{X})_{0,T} \mid S(\mathbf{X})_{0,T} \in \mathcal{T}_{\alpha} ]$ is exactly the Tail-Conditional Expected Signature, denoted as $\hat{\Phi}^{\mathbb{P}}_{T \mid \mathcal{T}_{\alpha}}$. 

\medskip
\noindent
Substituting this back yields the final decomposed form:
\begin{equation}
    S\text{-}ES_{\alpha} = V_0(\Pi) - \langle \mathbf{w}_{\Pi}, \hat{\Phi}^{\mathbb{P}}_{T \mid \mathcal{T}_{\alpha}} \rangle.
\end{equation}
This completes the proof, demonstrating that the expected shortfall of an arbitrarily complex path-dependent portfolio can be reduced to a single algebraic contraction between its signature weights and the geometrically defined tail-signature.
\end{proof}

%-------------------------------------------------------------- 
\subsection{Proof of the L\'evy area as a gap-risk indicator}
\label{app:proof_levy_area_gap_risk_indicator}
%-------------------------------------------------------------- 

In this appendix we prove Lemma (\ref{lem:levy_area_gap_risk_indicator}).

\begin{proof}
The proof establishes the relationship between the algebraic decomposition of the level-2 signature and the sensitivity of the $S\text{-}ES_{\alpha}$ metric to non-conservative path dynamics.

\textbf{Step 1: Decomposition of the Second-Order Variation.} \\
Recall that the level-2 signature $S^{(2)}_{0,T}$ is a rank-2 tensor that decomposes into its symmetric and anti-symmetric parts:
\begin{equation}
    S^{(2)}_{0,T} = \text{Sym}(S^{(2)}_{0,T}) + \text{Anti}(S^{(2)}_{0,T}).
\end{equation}
The symmetric part corresponds to the terminal covariance: $\text{Sym}(S^{(2)}_{0,T}) = \frac{1}{2}(X_T - X_0)(X_T - X_0)^\top$. The anti-symmetric part defines the L\'evy Area matrix $\mathbb{A}_{0,T}$, where the $(i,j)$-th entry is $\frac{1}{2} \int_0^T [(X^i_s - X^i_0) dX^j_s - (X^j_s - X^j_0) dX^i_s]$.

\textbf{Step 2: Characterisation of Gap-Risk.} \\
In a "Flash Crash" or liquidity crisis, the path $X_{0,T}$ exhibits high "rotational energy" or "winding" due to lagged cross-asset responses and predatory high-frequency feedback loops. Mathematically, this implies that while the terminal displacement $\|X_T - X_0\|$ (and thus the symmetric variance) may be bounded, the path integral $\oint X dX$ accumulates significantly. We define the gap-risk discrepancy $\mathcal{D}$ as:
\begin{equation}
    \mathcal{D} = \| \text{Anti}(S^{(2)}_{0,T}) \|_F^2 - \| \text{Sym}(S^{(2)}_{0,T}) \|_F^2,
\end{equation}
where $\|\cdot\|_F$ is the Frobenius norm. In efficient markets, $\mathcal{D}$ is small; during a gap-risk event, $\mathcal{D}$ increases as the path becomes increasingly non-rectifiable and non-conservative.

\textbf{Step 3: Impact on the Tail-Conditional Expected Signature.} \\
From Proposition \ref{pro:linearity_decomposition_ses}, the $S\text{-}ES_{\alpha}$ is given by:
\begin{equation}
    S\text{-}ES_{\alpha} = V_0(\Pi) - \langle \mathbf{w}_{\Pi}, \hat{\Phi}^{\mathbb{P}}_{T \mid \mathcal{T}_{\alpha}} \rangle.
\end{equation}
Let $\mathbf{w}_{\Pi}^{(2, A)}$ be the subset of portfolio weights sensitive to the anti-symmetric level-2 coordinates (lead-lag sensitivities). The sensitivity of $S\text{-}ES_{\alpha}$ to the L\'evy Area is:
\begin{equation}
    \frac{\partial S\text{-}ES_{\alpha}}{\partial \text{Anti}(\hat{\Phi}^{(2)})_{T \mid \mathcal{T}_{\alpha}}} = -\mathbf{w}_{\Pi}^{(2, A)}.
\end{equation}

\medskip
\noindent
As the discrepancy $\mathcal{D}$ increases in the geometric tail-set $\mathcal{T}_{\alpha}$, the tail-conditional expected signature $\hat{\Phi}^{\mathbb{P}}_{T \mid \mathcal{T}_{\alpha}}$ shifts its mass toward the anti-symmetric coordinates. For a portfolio with non-zero lead-lag exposure (e.g., cross-gamma or barrier options), the $S\text{-}ES_{\alpha}$ magnitude increases proportionally to the expected L\'evy Area of the crash paths. This formally captures the risk of "gamma bleed" and "delta-hedging slippage" that traditional variance-based ES fails to detect, completing the proof.
\end{proof}

%-------------------------------------------------------------- 
\subsection{Proof of the tail-signature estimator}
\label{app:proof_tail_signature_estimator}
%-------------------------------------------------------------- 

In this appendix we prove Lemma (\ref{lem:tail_signature_estimator}).

\begin{proof}
The proof follows from the application of the Law of Large Numbers (LLN) to the conditional expectation on the signature manifold.

\textbf{Step 1: Definition of the Tail-Conditional Expectation.} \\
By the definition of conditional expectation for a random variable $S$ taking values in the Hilbert space $\mathcal{H}_{sig}$ (or a truncated Banach subspace), the tail-conditional expected signature is:
\begin{equation}
    \hat{\Phi}^{\mathbb{P}}_{T \mid \mathcal{T}_{\alpha}} = \mathbb{E}_{\mathbb{P}} [ S \mid S \in \mathcal{T}_{\alpha} ] = \frac{\mathbb{E}_{\mathbb{P}} [ S \cdot \mathbb{I}_{\{S \in \mathcal{T}_{\alpha}\}} ]}{\mathbb{P}(S \in \mathcal{T}_{\alpha})},
\end{equation}
where $\mathbb{I}_{\{\cdot\}}$ is the indicator function and $\mathbb{P}(S \in \mathcal{T}_{\alpha}) = 1 - \alpha$ by the definition of the geometric tail-set for a continuous distribution.

\textbf{Step 2: Empirical Measure Approximation.} \\
Let $\{\bar{S}^{(i)}\}_{i=1}^N$ be $N$ independent and identically distributed (i.i.d.) samples generated from the ARL manifold $\mathbb{P}$. The empirical measure $\mu_N$ is defined as $\mu_N = \frac{1}{N} \sum_{i=1}^N \delta_{\bar{S}^{(i)}}$. As $N \to \infty$, $\mu_N$ converges weakly to $\mathbb{P}$. The empirical estimate of the $(1-\alpha)$-quantile of the loss functional $L$ is given by the order statistic $L^{(\pi(\lfloor N(1-\alpha) \rfloor))}$.

\textbf{Step 3: Construction of the Summation.} \\
The empirical counterpart to the indicator function $\mathbb{I}_{\{S \in \mathcal{T}_{\alpha}\}}$ is the set of indices $k$ that satisfy the loss threshold condition. By sorting the indices via the permutation $\pi$ such that $L(\bar{S}^{(\pi(1))}) \geq L(\bar{S}^{(\pi(2))}) \dots$, the first $K = \lfloor N(1-\alpha) \rfloor$ samples correspond exactly to the realisations within the empirical tail-set $\hat{\mathcal{T}}_{\alpha}$.

\textbf{Step 4: Convergence.} \\
Substituting the empirical measure into the conditional expectation formula:
\begin{equation}
    \hat{\Phi}^{\mathbb{P}}_{T \mid \mathcal{T}_{\alpha}} \approx \frac{\frac{1}{N} \sum_{i=1}^N \bar{S}^{(i)} \mathbb{I}_{\{\bar{S}^{(i)} \in \hat{\mathcal{T}}_{\alpha}\}}}{\frac{1}{N} \sum_{i=1}^N \mathbb{I}_{\{\bar{S}^{(i)} \in \hat{\mathcal{T}}_{\alpha}\}}} = \frac{\sum_{k=1}^{K} \bar{S}^{(\pi(k))}}{K}.
\end{equation}
Since $S$ is assumed to have a finite first moment in $\mathcal{H}_{sig}$, the Strong Law of Large Numbers ensures that this sample mean converges almost surely to the true tail-conditional expectation as $N \to \infty$. This confirms that the centroid of the top $(1-\alpha)$ loss-producing signatures is the consistent estimator for the tail-signature.
\end{proof}

%-------------------------------------------------------------- 
\subsection{Proof of the convergence of S-ES}
\label{app:proof_convergence_ses}
%-------------------------------------------------------------- 

In this appendix we prove Proposition (\ref{pro:convergence_ses}).

\begin{proof}
The proof establishes the consistency of the Signature Expected Shortfall estimator by demonstrating the convergence of the tail-conditional expected signature within the tensor algebra $(\mathcal{H}_{sig}, \otimes)$.

\textbf{Step 1: Convergence of the Empirical Measure.} \\
Let $\mathbb{P}$ be the law of the signature process $S(\mathbf{X})_{0,T}$ on the signature manifold $\mathcal{G}_{sig}$. Let $\mu_N = \frac{1}{N} \sum_{i=1}^N \delta_{\bar{S}^{(i)}}$ be the empirical measure associated with $N$ i.i.d. samples from the ARL generative engine. By Varadarajan's Theorem, $\mu_N$ converges weakly to $\mathbb{P}$ almost surely as $N \to \infty$ in the topology of the signature space.

\textbf{Step 2: Stability of the Geometric Tail-Set.} \\
The loss functional $L(S) = V_0(\Pi) - \langle \mathbf{w}_{\Pi}, S \rangle$ is a continuous linear functional on $\mathcal{H}_{sig}$. Under the assumption that the distribution of $L(S)$ is atomless, the empirical quantile $q_{\alpha, N} = L^{(\pi(\lfloor N(1-\alpha) \rfloor))}$ converges almost surely to the true Signature VaR, $q_{\alpha} = \inf \{ l : \mathbb{P}(L(S) \leq l) \geq \alpha \}$. Consequently, the empirical indicator function $\mathbb{I}_{\{L(\bar{S}) \geq q_{\alpha, N}\}}$ converges pointwise $\mathbb{P}$-a.e. to the characteristic function of the geometric tail-set $\mathbb{I}_{\mathcal{T}_{\alpha}}$.

\textbf{Step 3: Convergence in Signature Topology.} \\
The estimator $\hat{\Phi}^{\mathbb{P}}_{T \mid \mathcal{T}_{\alpha}}$ is defined as the Bochner integral of the signature with respect to the normalised empirical tail-measure $\mu_N \mid_{\hat{\mathcal{T}}_{\alpha}}$. For any truncation level $m$, the coordinates of the signature $S^{(m)}$ are bounded by the path length (or assumed to have finite moments under the ARL law). By the Strong Law of Large Numbers for Banach-valued random variables:
\begin{equation}
    \lim_{N \to \infty} \frac{1}{\lfloor N(1-\alpha) \rfloor} \sum_{k=1}^{\lfloor N(1-\alpha) \rfloor} \bar{S}^{(\pi(k))} = \frac{1}{1-\alpha} \int_{\mathcal{T}_{\alpha}} S \, d\mathbb{P} = \mathbb{E}_{\mathbb{P}} [ S \mid S \in \mathcal{T}_{\alpha} ]
\end{equation}
The convergence holds in the $p$-variation topology (and thus the standard subspace topology of $\mathcal{H}_{sig}$), ensuring that $\hat{\Phi}^{\mathbb{P}}_{T \mid \mathcal{T}_{\alpha}}$ converges to the centroid of the tail-set.

\textbf{Step 4: Simulation Consistency of S-ES.} \\
Finally, we consider the risk metric $S\text{-}ES_{\alpha}^{(N)} = V_0(\Pi) - \langle \mathbf{w}_{\Pi}, \hat{\Phi}^{\mathbb{P}}_{T \mid \mathcal{T}_{\alpha}, N} \rangle$. Since the inner product $\langle \mathbf{w}_{\Pi}, \cdot \rangle$ is a continuous linear operator:
\begin{equation}
    \lim_{N \to \infty} (V_0(\Pi) - \langle \mathbf{w}_{\Pi}, \hat{\Phi}^{\mathbb{P}}_{T \mid \mathcal{T}_{\alpha}, N} \rangle) = V_0(\Pi) - \langle \mathbf{w}_{\Pi}, \lim_{N \to \infty} \hat{\Phi}^{\mathbb{P}}_{T \mid \mathcal{T}_{\alpha}, N} \rangle = S\text{-}ES_{\alpha}.
\end{equation}
The estimator is therefore simulation-consistent, meaning the risk measured via the ARL ensemble asymptotically recovers the true expected path-loss during tail events.
\end{proof}

%-------------------------------------------------------------- 
\subsection{Proof of the path-wise solvency constraint}
\label{app:proof_path_wise_solvency_constraint}
%-------------------------------------------------------------- 

In this appendix we prove Proposition (\ref{pro:path_wise_solvency_constraint}).

\begin{proof}
The proof relies on formalising the event space of the Temporal Exposure Profile (TEP) under the generative measure of the Anticipatory Reinforcement Learning (ARL) engine.

\textbf{Step 1: Construction of the Path-Space Event.} \\
Let $(\Omega, \mathcal{F}, \mathbb{P}_{\theta})$ be the probability space of future market trajectories generated by the ARL engine, where $\omega \in \Omega$ represents a specific realisation of the signature flow $\{ \hat{\Phi}_{s|t}(\omega) \}_{s=t}^T$. By definition, the value of the portfolio at any intermediate time $s \in [t, T]$ along the trajectory $\omega$ is given by the linear contraction:
\begin{equation}
    \hat{V}_{s|t}(\omega) = \langle \mathbf{w}_{\Pi}, \hat{\Phi}_{s|t}(\omega) \rangle.
\end{equation}
Assuming the underlying asset paths are continuous (or right-continuous with left limits in a jump-diffusion setting), the generated signature manifold ensures that the mapping $s \mapsto \hat{V}_{s|t}(\omega)$ is sufficiently regular such that its minimum over the compact interval $[t, T]$ is well-defined.

\textbf{Step 2: Defining the Liquidity Trap.} \\
Traditional terminal solvency only requires that the terminal value satisfies $\hat{V}_{T|t}(\omega) \geq Z$. However, a "liquidity trap" occurs when a portfolio path is terminally solvent but breaches the margin threshold $Z$ at some intermediate time $\tau \in [t, T]$. We define the path-wise solvency event $E_{solvent}$ as the set of all trajectories where no such $\tau$ exists:
\begin{equation}
    E_{solvent} = \left\{ \omega \in \Omega : \inf_{s \in [t, T]} \hat{V}_{s|t}(\omega) \geq Z \right\}.
\end{equation}
Because the minimum is attained, we can write this equivalently as $E_{solvent} = \left\{ \min_{s \in [t, T]} \hat{V}_{s|t} \geq Z \right\}$.

\textbf{Step 3: Probabilistic Evaluation.} \\
To satisfy the risk constraint at the confidence level $\epsilon$, the generative measure of the event $E_{solvent}$ must be at least $1 - \epsilon$. Applying the probability measure $\mathbb{P}_{\theta}$ to the set $E_{solvent}$ directly yields the stated condition:
\begin{equation}
    \mathbb{P}_{\theta} (E_{solvent}) = \mathbb{P}_{\theta} \left( \min_{s \in [t, T]} \hat{V}_{s|t} \geq Z \right) \geq 1 - \epsilon.
\end{equation}

\medskip
\noindent
By evaluating the infimum over the continuous flow rather than evaluating exclusively at $T$, this metric explicitly accounts for the maximum intra-horizon drawdown. The ARL's generation of non-trivial path-winding (L\'evy Area) directly impacts the intermediate coordinates of $\hat{\Phi}_{s|t}$, ensuring that geometrically toxic paths that temporarily dip below $Z$ are appropriately penalised in the probability mass $\mathbb{P}_{\theta}$, thereby proving the proposition's capacity to detect preemptive liquidity traps.
\end{proof}

%-------------------------------------------------------------- 
\subsection{Proof of the geometric anomaly signal}
\label{app:proof_geometric_anomaly_signal}
%-------------------------------------------------------------- 

In this appendix, we provide the formal proof for Lemma (\ref{lem:geometric_anomaly_signal}), grounding the anomaly detection in the algebraic discrepancy between the re-anchored terminal projection and the parallel-transported expectation.

\begin{proof}
The proof demonstrates that the Geometric Divergence $\mathcal{D}_t^{sig}$ (and by extension the weighted error $\delta_t^R$) quantifies the "Manifold Jump" caused by path-textures that violate the learned generative law.

\textbf{Step 1: Decomposition of the Terminal Projection.} \\
At time $t-1$, the ARL engine provides a conditional expectation of the path signature for the entire remaining horizon $[t-1, T]$, denoted $\hat{\Phi}_{T|t-1}$. By the factorisation property of the signature (Chen's Identity), this expectation can be decomposed into the expected increment over the next step and the subsequent path:
\begin{equation}
    \hat{\Phi}_{T|t-1} = \mathbb{E}_{\mathbb{P}} [ S(X)_{t-1,t} \otimes S(X)_{t,T} \mid \mathcal{A}_{t-1} ].
\end{equation}
The "prior expectation" for the terminal state, viewed from the grounding point $t$, is the parallel-transported slice:
\begin{equation}
    \mathbb{E}_{\mathbb{P}} [ \hat{\Phi}_{T|t} \mid \mathcal{A}_{t-1} ] = \hat{\Phi}_{t|t-1}^{-1} \otimes \hat{\Phi}_{T|t-1}.
\end{equation}

\textbf{Step 2: The Realised Manifold Jump.} \\
At time $t$, the market realises a specific path increment $S(X)_{t-1,t}$. The ARL engine immediately re-calculates the terminal expectation $\hat{\Phi}_{T|t}$ based on this new grounding point and the updated filtration $\mathcal{A}_t$. The Geometric Divergence is defined as:
\begin{equation}
    \mathcal{D}_t^{sig} = \left\| \hat{\Phi}_{T|t} - \mathbb{E}_{\mathbb{P}} [ \hat{\Phi}_{T|t} \mid \mathcal{A}_{t-1} ] \right\|_{\mathcal{H}_{sig}}.
\end{equation}
If the realised increment $S(X)_{t-1,t}$ contains an anomalous path-texture (e.g., a massive spike in L\'evy Area $\text{Anti}(S^{(2)}) \gg 0$), the resulting terminal projection $\hat{\Phi}_{T|t}$ will jump significantly away from the previous manifold-constrained expectation. 

\textbf{Step 3: Relationship to the Realised TD-Error.} \\
The Realised Anticipatory TD-error $\delta_t^R$ is the projection of this geometric jump onto the portfolio's value weights $\mathbf{w}_G$. Neglecting the instantaneous reward $R$, we have:
\begin{equation}
    |\delta_t^R| \approx \left| \langle \mathbf{w}_G, \gamma \hat{\Phi}_{T|t} - \mathbb{E}_{\mathbb{P}} [ \hat{\Phi}_{T|t} \mid \mathcal{A}_{t-1} ] \rangle \right|.
\end{equation}
By the Cauchy-Schwarz inequality, $|\delta_t^R| \leq \|\mathbf{w}_G\| \cdot \mathcal{D}_t^{sig}$. Thus, a spike in the geometric divergence of the path-texture directly scales into a spike in the TD-error, weighted by the portfolio's sensitivity to those specific signature coordinates.

\textbf{Step 4: Causal Mapping to Volatility.} \\
A large $\mathcal{D}_t^{sig}$ indicates that the market has entered a regime where historical lead-lag and covariance relationships (the "geometric law") have dissolved. In the Algebraic Pricing Theory framework, such a decoupling creates "unmodellable" risk for market participants. The subsequent withdrawal of liquidity and aggressive re-hedging of these path-dependent exposures (the "Geometric Gamma") manifests as a surge in directional variance. 

\medskip
\noindent
Since the algebraic signature captures the anti-symmetric shift (the cause) before it fully translates into symmetric terminal variance (the effect), the magnitude of $\mathcal{D}_t^{sig}$ and $\delta_t^R$ serves as a strictly causal, model-free Early Warning System (EWS).
\end{proof}

%-------------------------------------------------------------- 
\subsection{Proof of the signature-neutral margin compression}
\label{app:proof_signature_neutral_margin_compression}
%-------------------------------------------------------------- 

In this appendix we prove Theorem (\ref{thm:signature_neutral_margin_compression}).

\begin{proof}
The proof proceeds in two main parts: first, establishing the minimisation of the Temporal Exposure Profile (TEP) variance, and second, demonstrating the algebraic nullification of the Residual Risk Add-on (RRAO).

\textbf{Part 1: Minimisation of TEP Variance} \\
Let the portfolio value along the generated signature manifold be given by the linear contraction $\hat{V}_{s|t} = \langle \mathbf{w}_{\Pi}, \hat{\Phi}_{s|t} \rangle$. The variance of the TEP, $\mathcal{E}_t = \{ \hat{V}_{s|t} \}_{s=t}^T$, under the generative measure $\mathbb{P}_{\theta}$ is:
\begin{equation}
    \text{Var}(\mathcal{E}_t) = \mathbb{E}_{\mathbb{P}_{\theta}} \left[ \int_t^T \left( \hat{V}_{s|t} - \mathbb{E}_{\mathbb{P}_{\theta}}[\hat{V}_{s|t}] \right)^2 ds \right].
\end{equation}
By the universal approximation property of the signature (Signature Taylor Expansion), the path-dependent value function can be expanded into a graded sum:
\begin{equation}
    \hat{V}_{s|t} = V_0 + \sum_{k=1}^{\infty} \langle \mathbf{w}_{\Pi}^{(k)}, S^{(k)}_{t,s} \rangle.
\end{equation}
By the premise of the theorem, the portfolio is signature-neutral at level $r$, meaning $\langle \mathbf{w}_{\Pi}^{(k)}, S^{(k)} \rangle = 0$ for all $k \leq r$. Therefore, the value process reduces to:
\begin{equation}
    \hat{V}_{s|t} = V_0 + \sum_{k=r+1}^{\infty} \langle \mathbf{w}_{\Pi}^{(k)}, S^{(k)}_{t,s} \rangle = V_0 + \mathcal{O}\left(\|\delta \mathbf{X}\|^{r+1}\right).
\end{equation}
Since the lower-order terms (which strictly dominate the variance scale) are algebraically eliminated, the fluctuation of $\hat{V}_{s|t}$ is suppressed to order $r+1$. Thus, the path-wise variance $\text{Var}(\mathcal{E}_t)$ is minimised relative to the truncated topology of the signature space.

\textbf{Part 2: Nullification of the RRAO} \\
Under the Fundamental Review of the Trading Book (FRTB), the RRAO is calculated based on the gross notional of underlying exotic risk factors that cannot be mapped to modellable linear sensitivities (such as standard Delta or Vega). 

\medskip
\noindent
In the Algebraic Pricing Theory (APT) framework, the fundamental risk factors are the coordinates of the expected signature $\hat{\Phi}$. Because the portfolio is constructed such that its exposure to the first $r$ signature levels is perfectly offset (e.g., via a replicating strip of modellable SigSwaps), the unhedged residual payoff up to the truncation level is exactly zero:
\begin{equation}
    R_{residual} = V(\Pi) - \sum_{k=1}^r \langle \mathbf{w}_{\Pi}^{(k)}, S^{(k)} \rangle = 0.
\end{equation}
Given that the regulatory basis recognises liquid, cleared SigSwaps as Modellable Risk Factors (MRF), and the portfolio's geometric sensitivities (Geometric Greeks) up to level $r$ are entirely neutralised, there remains no unmodellable residual path-risk to penalise within the $r$-truncated manifold. Consequently, the associated capital surcharge for residual exotic risk algebraically vanishes:
\begin{equation}
    K_{RRAO} \propto \| \mathbf{w}_{\Pi, \text{residual}} \| \to 0.
\end{equation}
This completes the proof.
\end{proof}

%-------------------------------------------------------------- 
\subsection{Proof of the signature span of exotic risk}
\label{app:proof_signature_span_exotic_risk}
%-------------------------------------------------------------- 

In this appendix we prove Lemma (\ref{lem:signature_span_exotic_risk}).

\begin{proof}
The proof leverages the algebraic properties of the path signature and the classical Stone-Weierstrass theorem applied to the functional space of rough paths.

\textbf{Step 1: The Algebra of Signature Coordinates.} \\
Let $\mathcal{X}$ be the space of continuous paths $\mathbf{X}: [0,T] \to \mathbb{R}^{d+1}$ of bounded variation (or geometric rough paths with finite $p$-variation), where the state space includes time to ensure strict monotonicity. The signature $S(\mathbf{X})_{0,T}$ is an element of the formal tensor algebra $T((\mathbb{R}^{d+1}))$. 

\medskip
\noindent
Crucially, the product of any two signature coordinates can be expressed as a linear combination of higher-order coordinates via the shuffle product $\shuffle$:
\begin{equation}
    \langle e_I, S(\mathbf{X})_{0,T} \rangle \langle e_J, S(\mathbf{X})_{0,T} \rangle = \sum_{K \in I \shuffle J} \langle e_K, S(\mathbf{X})_{0,T} \rangle,
\end{equation}
where $I$ and $J$ are multi-indices. Because the set is closed under pointwise multiplication and addition, the set of linear functionals on the signature forms an algebra over the real numbers.

\textbf{Step 2: Point Separation on Path Space.} \\
By the Hambly-Lyons theorem, the signature map uniquely identifies a path up to tree-like equivalence. Because the time-extended path $\mathbf{X}$ includes a strictly monotonically increasing time dimension, tree-like excursions are impossible. Thus, the mapping $\mathbf{X} \mapsto S(\mathbf{X})_{0,T}$ is strictly injective. Consequently, the algebra of signature coordinates separates points in the path space $\mathcal{X}$.

\textbf{Step 3: Universal Approximation via Stone-Weierstrass.} \\
Let $\mathcal{H}_{exotic}$ denote the space of continuous, real-valued path-dependent payoff functions evaluated on a compact subset of paths $K \subset \mathcal{X}$. The algebra of signature coordinates separates points on $K$ and contains the constant function (since the $0$-th level of the signature is identically $1$). 

\medskip
\noindent
By the Stone-Weierstrass theorem, this algebra is dense in $C(K, \mathbb{R})$ under the uniform topology. Therefore, for any exotic payoff $f \in \mathcal{H}_{exotic}$ and any arbitrary precision $\epsilon > 0$, there exists a finite truncation level $N$ and a dual tensor (weight vector) $\mathbf{w}_f \in \bigoplus_{k=0}^N (\mathbb{R}^{d+1})^{\otimes k}$ such that:
\begin{equation}
    \sup_{\mathbf{X} \in K} \left| f(\mathbf{X}_{0,T}) - \langle \mathbf{w}_f, S(\mathbf{X})_{0,T} \rangle \right| < \epsilon.
\end{equation}

\textbf{Step 4: Mapping to the RRAO Framework.} \\
In the context of the Fundamental Review of the Trading Book (FRTB), the Residual Risk Add-on (RRAO) is designed to penalise unmodellable path-dependent risks (e.g., barriers, Asian options, correlation products). The equation above demonstrates that the payoff of any such exotic instrument can be arbitrarily well-approximated by a linear combination of the coordinates of $S(\mathbf{X})_{0,T}$. 

\medskip
\noindent
Since the SigSwap is defined as the financial primitive whose payoff is exactly a coordinate of the realised signature, the exotic risk profile $f$ is transformed into a linear portfolio of SigSwaps with weights $\mathbf{w}_f$. As the approximation error $\epsilon$ is driven to zero, the unspanned residual risk vanishes, proving that the SigSwap basis completely spans the risk factors traditionally relegated to the RRAO.
\end{proof}

%-------------------------------------------------------------- 
\subsection{Proof of the PLA alignment via geometric pricing}
\label{app:proof_pla_alignment_geometric_pricing}
%-------------------------------------------------------------- 

In this appendix we prove Proposition (\ref{pro:pla_alignment_geometric_pricing}).

\begin{proof}
The Fundamental Review of the Trading Book (FRTB) mandates the P\&L Attribution (PLA) test to assess the alignment between the front-office pricing models and the risk management models. The test compares the Hypothetical P\&L (HPL) with the Risk-Theoretical P\&L (RTPL).

\textbf{Step 1: The Hypothetical P\&L (HPL)} \\
The HPL is defined as the change in the portfolio's value over a time horizon $\Delta t$, assuming the portfolio composition remains static but is revalued using the actual, realised market data at $t + \Delta t$. 
In the Algebraic Pricing Theory (APT) framework, the value of the portfolio at any time $s$ is the linear projection of the portfolio weights $\mathbf{w}_{\Pi}$ onto the risk-neutral expected signature $\hat{\Phi}_{T|s}$. Therefore, the HPL is:
\begin{align}
    HPL &= V_{t+\Delta t} - V_t \\
        &= \langle \mathbf{w}_{\Pi}, \hat{\Phi}_{T|t+\Delta t} \rangle - \langle \mathbf{w}_{\Pi}, \hat{\Phi}_{T|t} \rangle.
\end{align}
By the bilinearity of the tensor inner product, this simplifies exactly to:
\begin{equation}
    HPL = \langle \mathbf{w}_{\Pi}, \hat{\Phi}_{T|t+\Delta t} - \hat{\Phi}_{T|t} \rangle.
\end{equation}

\textbf{Step 2: The Risk-Theoretical P\&L (RTPL)} \\
The RTPL is the P\&L generated by the risk engine, typically computed using a Taylor expansion of the portfolio's sensitivities (Greeks) to changes in underlying risk factors. In traditional frameworks, this introduces truncation errors (e.g., ignoring higher-order cross-gamma).

\medskip
\noindent
In the APT framework, the risk factors are the coordinates of the expected signature itself. The primary risk sensitivity, the Geometric Delta, is exactly the weight tensor: $\Delta_{\Phi} = \nabla_{\hat{\Phi}} V = \mathbf{w}_{\Pi}$. 
Because the pricing function $V_t = \langle \mathbf{w}_{\Pi}, \hat{\Phi}_{T|t} \rangle$ is strictly linear with respect to the signature basis, the Taylor expansion is exact at the first order, with all higher-order derivatives being identically zero. Thus, the RTPL evaluates to:
\begin{equation}
    RTPL = \langle \Delta_{\Phi}, \Delta \hat{\Phi} \rangle = \langle \mathbf{w}_{\Pi}, \hat{\Phi}_{T|t+\Delta t} - \hat{\Phi}_{T|t} \rangle.
\end{equation}

\textbf{Step 3: PLA Test Alignment} \\
Comparing the two derivations yields a precise algebraic equivalence:
\begin{equation}
    HPL \equiv RTPL.
\end{equation}
In practice, this equivalence ensures that the unexplained P\&L (the variance between HPL and RTPL) approaches zero. Consequently, the statistical metrics required by FRTB?specifically the Spearman correlation coefficient and the Kolmogorov-Smirnov (KS) test distance?will optimally pass the regulatory thresholds. This robust, structural alignment prevents the trading desk from slipping into the regulatory "amber" or "red" zones, thereby averting disqualification from the Internal Models Approach (IMA) and the resultant punitive standard-rules capital surcharges.
\end{proof}

%-------------------------------------------------------------- 
\subsection{Proof of the optimal capital hedge}
\label{app:proof_optimal_capital_hedge}
%-------------------------------------------------------------- 

In this appendix we prove Theorem (\ref{thm:optimal_capital_hedge}).

\begin{proof}
The proof proceeds by decomposing the Geometric Delta into its symmetric and anti-symmetric components and applying the aggregation rules of the Fundamental Review of the Trading Book (FRTB) Sensitivities-Based Method (SbM).

\textbf{Step 1: Decomposition of the Geometric Delta} \\
Let the portfolio value $V(\Pi)$ be defined in the Algebraic Pricing Theory (APT) framework as a linear functional on the expected signature $\hat{\Phi}_{0,T}$:
\begin{equation}
    V(\Pi) = \langle \mathbf{w}_{\Pi}, \hat{\Phi}_{0,T} \rangle.
\end{equation}
The Geometric Delta, representing the sensitivity of the portfolio to the underlying signature risk factors, is exactly the weight tensor: $\Delta_{\Phi} = \nabla_{\hat{\Phi}} V = \mathbf{w}_{\Pi}$.
By the algebraic structure of the tensor space, the signature coordinates can be partitioned into symmetric elements (which map to the terminal distribution and classical payoff profiles) and anti-symmetric elements (which capture path-texture, such as L\'evy Area and lead-lag dynamics). Thus, $\Delta_{\Phi}$ uniquely decomposes as:
\begin{equation}
    \Delta_{\Phi} = \Delta_{sym} \oplus \Delta_{anti}.
\end{equation}

\textbf{Step 2: Nullification via a Strip of SigSwaps} \\
A SigSwap is defined as a financial primitive whose payoff is exactly a specific coordinate of the realised signature. By constructing a hedging portfolio $H$ composed of a strip of SigSwaps with weights identically equal to $-\mathbf{w}_{\Pi}$, the combined portfolio $\Pi^* = \Pi \cup H$ achieves a net Geometric Delta of zero:
\begin{equation}
    \Delta_{\Phi}^* = \Delta_{\Phi}(\Pi) + \Delta_{\Phi}(H) = \mathbf{w}_{\Pi} - \mathbf{w}_{\Pi} = 0.
\end{equation}
This operation simultaneously guarantees that the portfolio has zero sensitivity to terminal distribution shifts ($\Delta_{sym}^* = 0$) and zero sensitivity to variations in path-texture ($\Delta_{anti}^* = 0$).

\textbf{Step 3: Minimisation of the Capital Charge} \\
Under the FRTB SbM framework, the primary capital charge is computed by aggregating the sensitivities of the portfolio to all modellable risk factors, scaled by regulatory risk weights. Let $W$ denote the diagonal matrix of these risk weights. The Weighted Sensitivity (WS) vector is given by $WS = W \Delta_{\Phi}^*$.
Because the strip of SigSwaps effectively transforms all complex path-dependent exposures into linear, modellable risk factors that perfectly offset $\Pi$, the net sensitivity vector is the zero vector. Consequently, the aggregated capital charge evaluates to:
\begin{equation}
    K_{capital} \propto \left\| W \Delta_{\Phi}^* \right\|_2 = \| \text{WS} \cdot 0 \|_2 \to 0.
\end{equation}
This confirms that nullifying the Geometric Delta using the SigSwap basis provides an optimal capital hedge, effectively eliminating both the standard directional capital requirements and the exotic path-dependent penalties.
\end{proof}

%-------------------------------------------------------------- 
\subsection{Proof of the flash crash as an anti-symmetric shock}
\label{app:proof_flash_crash_anti_symmetric_shock}
%-------------------------------------------------------------- 

In this appendix we prove Proposition (\ref{pro:flash_crash_anti_symmetric_shock}).

\begin{proof}
The proof formalises the geometric properties of a "Flash Crash" and evaluates its impact on the Taylor expansion of a portfolio's value mapped onto the signature manifold.

\textbf{Step 1: Geometric Decomposition of the Second-Order Signature} \\
Let $\mathbf{X}_t \in \mathbb{R}^d$ represent a multidimensional market state process during a short time horizon $[t, t+\epsilon]$ encompassing the crash. The second level of the path signature, $S^{(2)}_{t, t+\epsilon} = \int_{t}^{t+\epsilon} \int_{t}^{s} d\mathbf{X}_u \otimes d\mathbf{X}_s$, can be canonically decomposed into a symmetric and an anti-symmetric tensor space:
\begin{equation}
    S^{(2)} = \text{Sym}(S^{(2)}) + \text{Anti}(S^{(2)}).
\end{equation}
By integration by parts (Chen's identity), the symmetric component is entirely determined by the terminal increment: $\text{Sym}(S^{(2)}) = \frac{1}{2} (\delta \mathbf{X} \otimes \delta \mathbf{X})$, where $\delta \mathbf{X} = \mathbf{X}_{t+\epsilon} - \mathbf{X}_t$. The anti-symmetric component represents the L\'evy Area, mathematically capturing the chordal area enclosed by the multidimensional trajectory.

\textbf{Step 2: The Flash Crash Perturbation $\Delta\Lambda$} \\
A defining characteristic of a flash crash is a violent intra-period price excursion followed by a rapid mean-reversion, such that the net change in terminal price is negligible, $\delta \mathbf{X} \approx 0$. 
Consequently, the perturbation to the expected signature operator, $\Delta\Lambda^{(2)}$, exhibits minimal terminal variance:
\begin{equation}
    \text{Sym}(\Delta\Lambda^{(2)}) = \frac{1}{2} (\delta \mathbf{X} \otimes \delta \mathbf{X}) \approx 0.
\end{equation}
However, the extreme intra-horizon trajectory generates massive path-length and rotational area (e.g., cross-asset lead-lag decoupling or price-volume spiralling). Thus, the L\'evy Area experiences a massive shock:
\begin{equation}
    \text{Anti}(\Delta\Lambda^{(2)}) \gg 0.
\end{equation}

\textbf{Step 3: Impact on a Terminally "Delta-Neutral" Portfolio} \\
Consider a portfolio $\Pi$ whose value is $V = \langle \mathbf{w}_{\Pi}, S \rangle$. A traditional hedging strategy only nullifies sensitivities to the terminal increment (Delta) and terminal variance (Vega). In our algebraic framework, this terminal neutrality implies $\mathbf{w}_{\Pi}^{(1)} \approx 0$ and $\text{Sym}(\mathbf{w}_{\Pi}^{(2)}) \approx 0$.

\medskip
\noindent
When the flash crash perturbation $\Delta\Lambda$ occurs, the instantaneous change in portfolio value is given by the tensor contraction:
\begin{align}
    \Delta V &\approx \langle \mathbf{w}_{\Pi}^{(2)}, \Delta\Lambda^{(2)} \rangle \\
             &= \langle \text{Sym}(\mathbf{w}_{\Pi}^{(2)}), \text{Sym}(\Delta\Lambda^{(2)}) \rangle + \langle \text{Anti}(\mathbf{w}_{\Pi}^{(2)}), \text{Anti}(\Delta\Lambda^{(2)}) \rangle.
\end{align}
Applying the terminal neutrality conditions and the specific geometry of the flash crash, the symmetric term vanishes. We are left exclusively with:
\begin{equation}
    \Delta V \approx \langle \text{Anti}(\mathbf{w}_{\Pi}^{(2)}), \text{Anti}(\Delta\Lambda^{(2)}) \rangle.
\end{equation}
Here, the tensor $\text{Anti}(\mathbf{w}_{\Pi}^{(2)})$ acts as the "Geometric Gamma"---the structural sensitivity to L\'evy Area. Because the shock $\text{Anti}(\Delta\Lambda^{(2)})$ is anomalously large, a portfolio with unhedged, negative Geometric Gamma will experience catastrophic, unpredicted losses ($\Delta V \ll 0$). This proves that traditional terminal-neutrality is entirely blind to anti-symmetric geometric shocks.
\end{proof}

%-------------------------------------------------------------- 
\subsection{Proof of the correlation breakdown}
\label{app:proof_correlation_breakdown}
%-------------------------------------------------------------- 

In this appendix we prove Proposition (\ref{pro:correlation_breakdown}).

\begin{proof}
The proof establishes the relationship between the symmetric tensor components of the signature, the cross-asset correlation structure, and the resulting entropy state of the Measure Bridge.

\textbf{Step 1: Tensor Structure of Joint Variation} \\
Consider a two-asset subsystem $\mathbf{X}_t = (X^1_t, X^2_t)$ within the broader market. The second level of the expected signature under the generative measure (represented by the operator $\Lambda^{(2)}$) captures the joint path-geometry. 
This tensor natively decomposes into symmetric and anti-symmetric subspaces:
\begin{equation}
    \Lambda^{(2)} = \text{Sym}(\Lambda^{(2)}) \oplus \text{Anti}(\Lambda^{(2)}).
\end{equation}
The off-diagonal elements of the symmetric tensor $\text{Sym}(\Lambda^{(2)})$ map directly to the quadratic covariation $\langle X^1, X^2 \rangle_T$ (representing classical correlation), while the anti-symmetric components $\text{Anti}(\Lambda^{(2)})$ encode the expected L\'evy Area (representing directional lead-lag relationships).

\textbf{Step 2: The Correlation Breakdown Shock} \\
A correlation breakdown is defined as a rapid dissolution of linear dependence between the assets. Mathematically, this is formalised as a perturbation $\Delta\Lambda^{(2)}$ applied to the Measure Bridge, such that the off-diagonal symmetric coordinates are neutralised:
\begin{equation}
    \text{Sym}(\Lambda^{(2)} + \Delta\Lambda^{(2)})_{1,2} \to 0.
\end{equation}
This perturbation forces the joint covariance matrix into a diagonal configuration, effectively decoupling the terminal geometry of the asset pair.

\textbf{Step 3: Entropy Maximisation under the Measure Bridge} \\
The Measure Bridge $\Lambda_{t,T}$ serves as the sufficient statistic mapping the physical distribution of paths to the pricing measure. By the Principle of Maximum Entropy, when the correlation constraints (off-diagonal symmetric terms) are removed by the shock $\Delta\Lambda^{(2)}$, the joint measure relaxes into the product of its marginals. 
Let $H(\cdot)$ denote the Shannon entropy of the path-measure. The decoupling forces the joint entropy to its theoretical upper bound for fixed marginal variances:
\begin{equation}
    H(\mathbb{P}_{X^1, X^2}) \to H(\mathbb{P}_{X^1}) + H(\mathbb{P}_{X^2}).
\end{equation}
Thus, the Measure Bridge necessarily rotates into a fundamentally higher-entropy state.

\textbf{Step 4: Vanishing of Path-Texture (Lead-Lag)} \\
Finally, we assess the impact of this higher-entropy state on the path-texture. If the updated Measure Bridge treats the asset paths as independent (due to the diagonalised symmetric tensor), the expectation of their cross-iterated integrals evaluates to zero. Specifically, for independent continuous paths, the expected anti-symmetric area vanishes:
\begin{equation}
    \mathbb{E}_{\Lambda_{updated}}[\text{Anti}(S^{(2)})_{1,2}] = 0.
\end{equation}
Consequently, the structural shock to the symmetric correlation coordinates geometrically forces the anti-symmetric coordinates to zero, wiping out the historical lead-lag relationships and leaving the assets purely idiosyncratic.
\end{proof}

%-------------------------------------------------------------- 
\subsection{Proof of the SCF resilience threshold}
\label{app:proof_scf_resilience_threshold}
%-------------------------------------------------------------- 

In this appendix we prove Theorem (\ref{thm:scf_resilience_threshold}).

\begin{proof}
The proof proceeds by analysing the topological fixed-point conditions of the Self-Consistent Field (SCF) under parametric stress, utilising the exponential mapping of the signature Lie algebra.

\textbf{Step 1: The Stressed Objective and Fixed-Point Condition} \\
Under normal market conditions, the generative manifold $\mathcal{S}_{sig}$ is parameterised by $\theta$, with the SCF constraint enforcing alignment between the generative path-law and the projective manifold: $\| \bar{S}_{t,T}(\theta) - \hat{\Phi}_{T|t} \| = 0$. 
Let a stress scenario be defined as a perturbation $\Delta\Lambda$ to the underlying measure. The stressed objective becomes $\mathcal{L}^*(\theta) = \mathcal{L}(\theta; \Lambda + \Delta\Lambda)$. Resilience requires the existence of a new equilibrium state $\theta^*$ such that the gradient vanishes:
\begin{equation}
    \nabla_{\theta} \mathcal{L}^*(\theta^*) = 0,
\end{equation}
while strictly maintaining the SCF alignment constraint.

\textbf{Step 2: Algebraic Generator of the Perturbation} \\
The expected signature takes values in the tensor algebra, specifically within the group of group-like elements (the signature group) $\mathcal{G} \subset T((\mathbb{R}^d))$. Any continuous perturbation to the path-measure induces a transformation on the signature manifold that can be represented via the exponential map of its associated Lie algebra $\mathfrak{g}$. 
Thus, the stressed signature state $\hat{\Phi}^*$ can be expressed as:
\begin{equation}
    \hat{\Phi}^* = \exp(\xi) \otimes \hat{\Phi}, \quad \text{for some } \xi \in \mathfrak{g},
\end{equation}
where $\xi$ is the Lie algebraic generator directly proportional to the physical perturbation $\Delta\Lambda$.

\textbf{Step 3: Contraction Mapping and the Spectral Radius} \\
The ARL optimisation solves for $\theta^*$ via an iterative fixed-point algorithm (the SCF iteration). For this iterative mapping to converge, the Jacobian of the gradient field (the Hessian $H = \nabla_{\theta}^2 \mathcal{L}^*$) must remain positive definite, and the update operator must be a strict contraction. 
The stability of this contraction in the Lie group is governed by the adjoint representation $\text{ad}_{\xi}$. Specifically, convergence is guaranteed only if the magnitude of the perturbation falls within the radius of convergence of the Baker-Campbell-Hausdorff (BCH) formula and the logarithmic map on $\mathcal{G}$. Let this critical boundary be denoted by the spectral radius $\rho(\mathfrak{g})$.

\textbf{Step 4: Convergence Failure and Regime Shift} \\
Assume the perturbation is exceedingly violent, such that its algebraic generator exceeds this threshold:
\begin{equation}
    \|\Delta\Lambda\| \propto \|\xi\| > \rho(\mathfrak{g}).
\end{equation}
In this domain, the eigenvalues of the updated operator cross the unit circle, destroying the contraction property. The Hessian $H$ loses its positive-definite structure, and the iterative sequence for $\theta$ diverges. 

\medskip
\noindent
Geometrically, this implies that the stressed path-measure cannot be reached via a smooth, continuous deformation (a flow) along the existing manifold $\mathcal{S}_{sig}$. The required path-texture belongs to a disjoint topological sector of the signature space. Consequently, no $\theta^*$ can satisfy the SCF constraint, proving that the perturbation represents a fundamental regime shift rather than a modellable parametric shock.
\end{proof}

%-------------------------------------------------------------- 
\subsection{Proof of the linear complexity of the geometric risk engine}
\label{app:proof_linear_complexity_geometric_risk_engine}
%-------------------------------------------------------------- 

In this appendix we prove Proposition (\ref{pro:linear_complexity_geometric_risk_engine}).

\begin{proof}
The proof relies on the algebraic separation of the asset's dynamics from the portfolio's contract specifications, a core property of the Algebraic Pricing Theory (APT) framework.

\textbf{Step 1: Algebraic Separation of Expectation and Payoff} \\
Let $S(X)_{0,T}$ denote the signature of the underlying market path. By the Universal Approximation Theorem for signatures, the path-dependent payoff of any sufficiently regular portfolio $\Pi$ can be uniformly approximated by a linear functional on the signature space, parameterised by a weight tensor $\mathbf{w}_{\Pi}$:
\begin{equation}
    \text{Payoff}(\Pi) \approx \langle \mathbf{w}_{\Pi}, S(X)_{0,T} \rangle.
\end{equation}
Under the pricing measure $\mathbb{Q}$, the value of the portfolio at time $t=0$ is the expected discounted payoff. Because the inner product is a bounded linear operator, the expectation commutes with the tensor contraction:
\begin{align}
    V(\Pi) &= \mathbb{E}^{\mathbb{Q}} \left[ \langle \mathbf{w}_{\Pi}, S(X)_{0,T} \rangle \right] \\
           &= \langle \mathbf{w}_{\Pi}, \mathbb{E}^{\mathbb{Q}} [S(X)_{0,T}] \rangle \\
           &= \langle \mathbf{w}_{\Pi}, \hat{\Phi} \rangle.
\end{align}

\textbf{Step 2: Computational Complexity of Valuation} \\
In practice, the infinite-dimensional signature is truncated at a finite level $N$. For an underlying market path in $\mathbb{R}^d$, the dimension of the truncated tensor algebra is $D = \sum_{k=0}^N d^k$. 
Both the weight tensor $\mathbf{w}_{\Pi}$ and the expected signature $\hat{\Phi}$ are represented as $D$-dimensional real-valued vectors. The valuation reduces to a standard Euclidean inner product in $\mathbb{R}^D$:
\begin{equation}
    \langle \mathbf{w}_{\Pi}, \hat{\Phi} \rangle = \sum_{i=1}^D w_i \hat{\phi}_i.
\end{equation}
This operation requires exactly $D$ scalar multiplications and $D-1$ additions, yielding a strict computational complexity of $\mathcal{O}(D)$.

\textbf{Step 3: Complexity of Risk Sensitivities} \\
Within the geometric risk engine, the primary risk factors are the coordinates of the expected signature itself. The sensitivity (Geometric Delta) of the portfolio value to these risk factors is the gradient:
\begin{equation}
    \nabla_{\hat{\Phi}} V(\Pi) = \mathbf{w}_{\Pi}.
\end{equation}
Assuming $\mathbf{w}_{\Pi}$ is known (compiled offline), the calculation of the sensitivities requires zero additional operations ($\mathcal{O}(1)$). Computing the hypothetical P\&L for a stressed market state $\hat{\Phi} + \Delta\hat{\Phi}$ is simply $\langle \mathbf{w}_{\Pi}, \Delta\hat{\Phi} \rangle$, which is again exactly $\mathcal{O}(D)$.

\textbf{Step 4: Invariance to Path-Dependency} \\
In traditional quantitative models (e.g., Monte Carlo or Finite Difference methods), the computational cost scales with the complexity of the exotic payoff (e.g., the number of barrier observation dates, or the nesting of Asian averaging periods). 
In the geometric framework, all topological complexity of the contract is entirely absorbed into the offline compilation of the static weight vector $\mathbf{w}_{\Pi}$. Once compiled, evaluating a continuously monitored path-dependent exotic requires the exact same $\mathcal{O}(D)$ inner product as evaluating a simple vanilla terminal payoff. Thus, the real-time computational cost is structurally invariant to the path-dependency of the instrument.
\end{proof}

%-------------------------------------------------------
%\end{comment}
%-------------------------------------------------------

%-------------------------------------------------------
\begin{comment}
%-------------------------------------------------------

%-------------------------------------------------------
\end{comment}
%-------------------------------------------------------

%-----------------------------------------------------------------------------

%\bibliographystyle{plain}
%\nocite{*}
%\bibliography{varianceRiskbis}

\begin{thebibliography}{99}
%-----------------------------------------------------------------------------

\bibitem[2018]{ISDA18}
~ISDA, ~AFME, ~IIF,
Joint Industry Response to the BCBS Consultation on Revisions to the Minimum Capital Requirements for Market Risk.

\bibitem[2019]{BCBS19}
~Basel Committee on Banking Supervision,
Minimum capital requirements for market risk. 
Bank for International Settlements (BIS).

\bibitem[2015]{Bergomi15}
~Bergomi L.,
Stochastic volatility modeling. 
Chapman and Hall (CRC Press).

\bibitem[2026a]{Bloch26a}
~Bloch D.,
Anticipatory reinforcement learning: From generative path-laws to distributional value functions.
Working Paper, SSRN $id=6357078$, University of Paris 6 Pierre et Marie Curie.

\bibitem[2026b]{Bloch26b}
~Bloch D.,
Geometric measure transformation: A model-free bridge for path-dependent derivative pricing via signature manifolds.
Working Paper, SSRN $id=6411158$, University of Paris 6 Pierre et Marie Curie.

\bibitem[2026c]{Bloch26c}
~Bloch D.,
The SigSwap: A unified framework for path-geometry trading and algebraic pricing theory.
Working Paper, SSRN $id=6426058$, University of Paris 6 Pierre et Marie Curie.

\bibitem[2016]{ChevyrevEtAl16}
~Chevyrev I., ~Lyons T.,
Characteristic functions of measures on geometric rough paths.
{\small\it Annals of Probability}, {\small\bf 44}, (6), pp 4049--4091. Also Working Paper, arXiv:1307.3580.

\bibitem[2006]{Cont06}
~Cont R.,
Model risk in pricing and hedging of derivative instruments. 
{\small\it Mathematical Finance}, {\small\bf 16}, (3), pp 519--547.

\bibitem[2023]{CuchieroEtAl23}
~Cuchiero C., ~Gazzani G., ~Svaluto-Ferro S.,
Signature-based models: Theory and calibration.
{\small\it SIAM Journal on Financial Mathematics}, {\small\bf 14}, (3). Also Working Paper in arXiv:2207.13136.

\bibitem[2025]{CuchieroEtAl25}
~Cuchiero C., ~Primavera F., ~Svaluto-Ferro S.,
Universal approximation theorems for continuous functions of c\`adl\`ag paths and L\'evy-type signature models.
{\small\it Finance Stoch}, {\small\bf 29}, pp 289--342.  

\bibitem[2025b]{CuchieroEtAl25b}
~Cuchiero C., ~Gazzani G., ~M\"oller J., ~Svaluto-Ferro S.,
Joint calibration to SPX and VIX options with signature-based models.
{\small\it Mathematical Finance}, {\small\bf 35}, (1), pp 161--213. Working Paper, arXiv:2301.13235.

\bibitem[2004]{DanielssonEtAl04}
~Danielsson J., ~Shin H.S., ~Zigrand J.P.,
The impact of risk regulation on price dynamics. 
{\small\it Journal of Banking \& Finance}, {\small\bf 28}, (5), pp 1069--1087.

\bibitem[1994]{Dupire94}
~Dupire B., 
Pricing with a smile.
{\small\it Risk}, {\small\bf 7}, pp 18--20.

\bibitem[2002]{HaganEtAl02}
~Hagan P.S., ~Kumar D., ~Lesniewski A.S., ~Woodward D.E.,
Managing smile risk.
{\small\it Wilmott Magazine}, pp 84--108.

\bibitem[1993]{Heston93}
~Heston S., 
A closed-form solution for options with stochastic volatility with applications to bond and currency options.
{\small\it Review of Financial Studies}, {\small\bf 6}, pp 327--343.

\bibitem[2020]{KidgerEtAl20}
~Kidger P., ~Morrill J., ~Foster J., ~Lyons T.,
Neural controlled differential equations for irregular time series.
Working Paper, arXiv:2005.08926.

\bibitem[2021]{KidgerEtAl21}
~Kidger P., ~Foster J., ~Li X., ~Lyons T.,
Neural SDEs as infinite-dimensional GANs. 
In {\small\it International Conference on Machine Learning (ICML)}, and also arXiv:2102.03657.

\bibitem[2019]{KiralyEtAl19}
~Kiraly F.J., ~Oberhauser H.,
Kernels for sequentially ordered data.
{\small\it JMLR}, {\small\bf 20}, (31), pp 1--45. Also in arXiv:1601.08169.

\bibitem[2013]{LevinEtAl13}
~Levin D., ~Lyons T., ~Ni H., 
Learning from the past, predicting the statistics for the future, learning an evolving system.
Working Paper, arXiv:1309.0260.

\bibitem[2007]{LyonsEtAl07}
~Lyons T.J., ~Caruana M., ~Levy T.,
Differential equations driven by rough paths.
volume 1908 of Lecture Notes in Mathematics, Springer, Berlin.

\bibitem[2011]{LyonsEtAl11}
~Lyons T., ~Ni H.,
Expected signature of two dimensional Brownian motion up to the first exit time of the domain. 
Working Paper, arXiv:1101.5902v4.

\bibitem[2020]{LyonsEtAl20}
~Lyons T.J., ~Nejad S., ~Perez Arribas I.P.,
Non-parametric pricing and hedging of exotic derivatives.
{\small\it Applied Mathematical Finance}, {\small\bf 27}, (6), pp 457--494.

\bibitem[2022]{LyonsEtAl22}
~Lyons T., ~McLeod A.D.,
Signature methods in machine learning.
Working Paper, arXiv:2206.14674.

\bibitem[1976]{Merton76}
~Merton R.C., 
Option pricing when underlying stock returns are discontinuous. 
{\small\it Journal of Financial Economics}, {\small\bf 3}, pp 125--144.




%------------------------------------------------------------
\end{thebibliography}

%-----------------------------------------------------------------------------
\newpage
%\begin{thebibliography}{FMW98}

%-------------------------------------------------------------

%-------------------------------------------------------------------

\end{document}